\documentclass[11pt,draftcls,onecolumn]{IEEEtran} 

\usepackage{cite}      

\usepackage{graphicx}  

\usepackage{psfrag}    

\usepackage{subfigure} 

\usepackage{url}       

\usepackage{stfloats}  

\usepackage{amsmath}   
\usepackage{amsmath,amssymb,graphicx,url,color,bbm, amsthm,amsfonts,dsfont,array,psfrag,stmaryrd,wasysym}
\usepackage{enumerate}
\usepackage[ruled,vlined]{algorithm2e}

\DeclareMathOperator{\supp}{supp}
\DeclareMathOperator{\rank}{rank}
\DeclareMathOperator{\prox}{prox}

\DeclareMathOperator{\diag}{diag}

\DeclareMathOperator*{\argmin}{{arg\,min}\ }

\makeatletter
\DeclareRobustCommand\onedot{\futurelet\@let@token\@onedot}
\def\@onedot{\ifx\@let@token.\else.\null\fi\xspace}

\def\eg{\emph{e.g}\onedot} 
\def\ie{\emph{i.e}\onedot}

\def\etal{\emph{et al}\onedot}

\def\card{\mathbf{Card}}
\def\prox{\textit{prox}}

\def\mhat{\widehat m}

\def\tcal{\mathcal T}
\def\acal{\mathcal A}
\def\atilde{\widetilde A}
\def\abf{\mathbf H}
\def\sbf{\mathbf S}

\def\sp{\,\,}

\def\mhat{\widetilde m}

\def\tcal{\mathcal T}
\def\acal{\mathcal A}
\def\atilde{\widetilde A}
\def\Hbar{\overline H}

\def\lm{\ell_{2}/\ell_{1}}

\def\usm[#1]{\atilde_{#1}}

\def\O{\mathcal O}

\def\na{n_{1}}
\def\nb{n_{2}}
\def\lm{\ell_{2}/\ell_{1}}

\newcommand{\defeq}{\triangleq}

\newcommand{\Rbb}{\mathbb{R}}
\newcommand{\Nbb}{\mathbb{N}}

\newcommand\LRstMixNormPPXA{\texttt{Nuclear-$\ell_{2}/\ell_{1}$}}
\newcommand\LRJS{\LRstMixNormPPXA}
\newtheorem{theorem}{Theorem}
\newtheorem{corollary}{Corollary}

\newtheorem{definition}{Definition}

\newtheorem{lemma}{Lemma}
\newtheorem{lemma*}{Lemma}

\newlength \myfigwidth
\newlength \myfigwidthll
\newlength \myfigwidthllI

\newlength \descriptionwidth

\if@twocolumn
\else
\fi
\if@twocolumn
\else
\fi

\if@twocolumn
\else
\fi

\begin{document}
%
\title{\LARGE{Compressed Sensing of Simultaneous Low-Rank and Joint-Sparse Matrices}}
%
%
\author{
Mohammad~Golbabaee,~\IEEEmembership{Student Member~IEEE},~and~Pierre~Vandergheynst
\thanks{The authors are with the Signal Processing Laboratory LTS2, Electrical Engineering Department, \'Ecole Polytechnique F\'ed\'erale de Lausanne (EPFL), Station 11, CH-1015 Lausanne, Switzerland.
This work was supported in part by the EU FET program through projects SMALL (FET-225913) and UNLocX (FET-255931), and the Swiss National Science Foundation under grant 200021-117884.
\protect\\
E-mail:\{mohammad.golbabaei,pierre.vandergheynst\}@epfl.ch. }}
\maketitle

\begin{abstract}

In this paper we consider the problem of recovering a high dimensional data matrix from a set of incomplete and noisy linear measurements. We introduce a new model 
that can efficiently restrict the degrees of freedom of the problem and is generic enough to find a lot of applications, for instance in multichannel signal compressed sensing (\eg sensor networks, hyperspectral imaging) and compressive sparse principal component analysis (s-PCA). We assume data matrices have a \emph{simultaneous low-rank and joint sparse} structure, and we propose a novel approach for efficient compressed sensing (CS) of such data.
 Our CS recovery approach is based on a convex minimization problem that incorporates this restrictive structure by jointly regularizing the solutions with their nuclear (trace) norm and $\lm$ mixed norm.
 Our theoretical analysis uses a new notion of restricted isometry property (RIP) and 
 shows that, for sampling schemes satisfying RIP, our approach can stably recover all low-rank and joint-sparse matrices. For a certain class of random sampling schemes satisfying a particular concentration bound (\eg the subgaussian ensembles) we derive a lower bound on the number of CS measurements indicating the near-optimality of our recovery approach as well as a significant enhancement 
 compared to the state-of-the-art.
 We introduce an iterative algorithm based on proximal calculus in order to solve the joint nuclear and $\lm$ norms minimization problem and, finally, we illustrate the empirical recovery phase transition of this approach by series of numerical experiments.

\end{abstract}
\section{Introduction}

Suppose we are given an $\na \times \nb$ matrix $X$ which is \emph{joint-sparse}, 
that is only $k\ll \na$ rows of this matrix contain nonzero entries (a similar argument holds for joint-column-sparsity), and in the meantime the underlying submatrix of nonzero rows (and equivalently the whole matrix $X$) has a \emph{low rank} $r \ll \min(k, \nb)$. Such matrices have few degrees of freedom: if one knows the indices of those $k$ nonzero rows, the corresponding submatrix will only have $(k+\nb-r)r$ degrees of freedom. 
Following the recent trend in the areas of compressed sensing \cite{donoho2006compressed, CRT-stable2005, candes2006compressive}, sparse approximation \cite{mallat1993matching,Chen98atomic} and low-rank matrix recovery \cite{lowrank, Matirxcompletion, NoisyMC, optspace}, we would like to know how many non-adaptive linear measurements of this matrix are sufficient to robustly represent the whole data and, more importantly, how this compares  to the state-of-the-art. 
%
%
The following questions are therefore natural:
\begin{itemize}
\item What are the efficient information/structure-preserving measurements? 
\item Is there any  computationally tractable algorithm that incorporates the simultaneous low-rank and joint-sparse structures into data recovery and does it perform better than the state-of-the-art?
\item And finally, how many measurements do we need to recover an exact low-rank and sparse matrix, and is the recovery algorithm \emph{stable} with respect to noisy measurements and matrices that are approximately low-rank or not exactly joint-sparse but \emph{compressible}? 
\end{itemize}
This paper attempts to answer these questions.
%
In this regard, we introduce a novel approach for CS reconstruction of simultaneous low-rank and joint-sparse matrices. Our approach is based on a convex problem that simultaneously uses low-rank and joint-sparsity promoting regularizations, namely, the \emph{nuclear} (\emph{trace}) norm and the $\lm$ mixed-norm of the solution matrix.
As a result, 
the solution of this problem tends to have both desired structures. Further, we propose an efficient algorithm in order to solve this optimization problem. 

Additionally, we establish a theoretical upper bound on the reconstruction error guaranteeing the \emph{stability} of our recovery approach against noisy measurements, \emph{non-exact} sparse and \emph{approximately} low-rank data matrices. We prove that, if the sampling scheme satisfies a \emph{specific} restricted isometry property (RIP), the proposed approach can stably recover \emph{all} joint-sparse and low-rank matrices. In particular and for a certain class of \emph{random} sampling schemes, we prove that this property implies the  following lower bound on the number $m$ of CS measurements:
\begin{eqnarray}\label{our bound}
m \gtrsim \mathcal{O} \big( k  \log(n_{1}/k) + k r  + n_{2} r \big).
\end{eqnarray}
Compared to the necessary condition for recovery of such structured matrices,  \ie $m>(k+\nb-r)r $, this measurement bound highlights the near-optimality of our recovery approach: its scaling order with respect to $k$ and $\nb$ has only an extra $\log$ factor due to the joint-sparsity pattern subset selection.

\subsection{Application in Multichannel Signal CS}\label{sec:app1}
Multichannel signals are undoubtedly a perfect fit for the new compressed sensing paradigm: on one hand, they have numerous applications \eg in computer vision, biomedical data, sensor networks, multispectral imaging and audio processing, just to name a few. 
On the other hand, the main present challenge is the tremendously large and increasing flow of data in those applications, which poses tight constraints on the available technologies at any of the acquisition, compression, transmission and data analysis steps. 
Despite their huge dimensionality, multichannel data are often highly redundant as a result of two categories of correlations:  \emph{inter-} and \emph{intra-}channel correlations. Any efficient dimensionality reduction scheme requires comprehensive modeling of both types of correlations. A straightforward extension of the classical single-channel CS idea \cite{donoho2006compressed, CRT-stable2005} only considers inter-channel signal correlations through an \emph{unstructured} sparsity assumption on the signal in each channel and the standard recovery approaches such as  $\ell_{1}$ minimization result in limited improvement.

The simultaneous low-rank and joint-sparse model turns out to be an ideal candidate to more restrictively model multichannel data and the underlying structures. In this case, the matrix $X$ represents a multichannel signal where $\nb$ denotes the number of channels and $\na$ is the signal dimensionality per channel. This model not only considers inter-channel signal correlations leading to a sparse representation, but also takes into account intra-channel correlations and is able to assign a particular structure for the sparsity~: signals of different channels share a \emph{unique} dominant sparsity pattern (support) and in addition, the values of those sparse coefficients are \emph{linearly dependent} across the multiple channels. Beside its restrictiveness, this model is generic enough to hold for a wide range of applications where the entire multichannel signal can be decomposed into few sparse sources through a \emph{linear source mixture model} as follows:
\[
X = \sbf \abf^{T}, 
\]   
where $\sbf$  is an $\na \times \rho$ matrix whose columns contain $\rho\ll \min(\na,\nb)$ source signals that are linearly mixed with an $\nb \times \rho$ matrix $\abf$ in order to generate the whole data across multiple channels. Indeed, for a sparse (w.l.o.g. in the canonical basis) source matrix $\sbf$ and an arbitrary mixing matrix $\abf$, this factorization  implies a low-rank ($r\leq \rho$) and joint-sparse $X$.

We have discussed in our previous articles \cite{LRJS, LRJS-HSI-ICASSP12, LRTV-HSI-ICIP12} the application of this model in compressive hyperspectral imagery. The linear mixture model typically holds for this type of images where $\abf$ contains the spectral signatures of a few materials present in the region whose distributions are indicated by the source images, \ie the columns of $\sbf$. We have shown that leveraging both low-rank and sparse structures enables CS to achieve significantly low sampling rates: \emph{A hyperspectral image can be subsampled by retaining only 3\% of its original size and yet be robustly reconstructed.}

Alternatively, the same 
model can be  deployed for a much more efficient CS acquisition in sensor network applications and  it can eventually enhance the tradeoff between the number $\nb$ of sensors and the complexity (\eg the number of measurements per channel $m/\nb$) of each sensor.
Indeed, typical sensor network setups are
 monitoring highly correlated phenomena which are mainly generated by a small number of (sparse) sources distributed in a region, and $\abf$ represents a possibly unknown Green's function accounting for the physical source-to-sensor signal propagation model.




\subsection{Compressive Sparse PCA}
Despite similarity in names, this paper considers a different problem than the robust principal component analysis (r-PCA) \cite{RPCA, C-RPCA}. We are not interested in decomposing a data matrix  into separate low-rank $L$ and sparse $S$ components, \ie $X = L+S$. 
In contrast, here $X$ simultaneously have both structures and our problem is more related to  the \emph{sparse principal component analysis} (s-PCA) literature \cite{SPCA-Lasso, SPCA-daspremont, SPCA-Bach}. The s-PCA problem considers the following factorization of a data matrix:
\[X \approx UV^{T}\]
where, the columns of $V \in \Rbb^{\rho\times \nb}$  contain the few $\rho \ll \min(\na, \nb)$ principal components (PCs) of $X$ and $U\in \Rbb^{\na \times \rho}$ is a sparse matrix of the corresponding loadings of the principal components. The classical PCA \cite{PCA}, which is a key tool for multivariate data analysis and dimensionality reduction, numerically involves a singular value decomposition with orthogonal matrices $U$ and $V$. 
However, the classical approach comes with a key drawback: 
in practical setups with $\na$ variables  and $\nb$ observations, most of the loadings are nonzero and thus, the main factors explaining data variance are linear combinations of \emph{all} variables, making it often difficult to interpret the results. The idea of s-PCA has been mainly proposed to address this shortcoming at the cost of applying more complex algorithms and loosing orthogonality of $U$ and $V$. 

As previously described in Section \ref{sec:app1}, this factorization (\ie having a small number $\rho$ of PCs with a sparse loading matrix $U$) implies a low-rank and joint-sparse structure for $UV^{T}$.
Therefore an alternative approach to solve s-PCA is to first use our method to approximate $X$ by a simultaneous low-rank and joint-sparse matrix, and then apply the conventional PCA to decompose the result into $U$ and $V$.
This approach (\ie s-PCA in the compressive domain) could be particularly useful for large-dimensional data analysis and visualization where data size is the main bottleneck. In addition, and to the best of our knowledge, the theoretical analysis provided in this paper is the first in this direction.

\subsection{Paper Organization}
This paper is organized as follows: in Section \ref{ch3-sec:priorarts} we review the most related works in  the compressed sensing literature and discuss the efficiency of different data priors for CS recovery. Section \ref{ch3-sec:rec-approach} presents our novel CS reconstruction scheme for simultaneous low-rank and joint-sparse data matrices. 
In Section \ref{ch3-sec:theory} we establish the theoretical guarantees of this recovery approach. The proofs of the main theorems are provided in Section \ref{sec:proofs}.
 In Section \ref{sec: PPXA} we introduce a set of iterative algorithms based on the proximal splitting method in order to solve the convex recovery problem. Finally, in Section \ref{ch2-sec:expe} we show the empirical recovery phase transition behavior of our method for different CS acquisition schemes and we compare them to prior art. 
 

\subsection{Notations}

Throughout this paper we frequently use the following notations:

For the matrix $X\in \Rbb^{\na \times \nb}$, $X_{i,.}$ and $X_{.,j}$ respectively denote its $i$th row and $j$th column, and $[X]_{i,j}$ is the element of $X$ in the $i$th row and the $j$th column. Moreover, $X_{vec}$ denotes the vector formed by stacking the columns of $X$ below one another:
 \[
 X_{vec} = \begin{bmatrix}
   X_{.,1}\\
   X_{.,2}\\
   \vdots\\
   X_{.,\nb}
     \end{bmatrix}.
 \]

The vector space of matrices is equipped with an inner product defined for any pair of matrices $X,X' \in \Rbb^{\na \times \nb}$ as follows:
\[ 
\left\langle X,X' \right\rangle \defeq \textbf{Tr}(X^{T}X')=\sum_{i,j} [X]_{i,j}[X']_{i,j},
\]
where, $\textbf{Tr}(.)$ is the trace of a square matrix. Moreover, four matrix norms that we frequently use are:
\begin{enumerate}
\item The spectral norm $\|X\|$: the largest singular value of $X$.
 \item The nuclear norm $\|X\|_{*}\defeq \sum_{i} \sigma_{i}$ which is the sum of the singular values ($\sigma_{i}$) of $X$.
\item The Frobenius norm $\|X\|_{F} \defeq \sqrt {\left\langle X,X\right \rangle }  = \sqrt {\sum_{i,j} [X]_{i,j}^{2}} = \sqrt{\sum_{i} \sigma_{i}^{2}}$.
 \item  The $\lm$ mixed-norm $\|X\|_{2,1} \defeq \sum_{i} \sqrt {\sum_{j} [X]^{2}_{i,j}}$ that is the sum of the $\ell_{2}$ norms of the matrix rows. 
\end{enumerate}

Since $\|.\|_{*}$ is a norm, it is \emph{subadditive} \ie, $\|X+X'\|_{*} \leq \|X\|_{*}+\|X'\|_{*}$  and the equality (\emph{additivity}) holds if $X$ and $X'$ have row and column spaces that are orthogonal  $X^{T}X'=X(X')^{T}=\mathbf{0}$ \cite{Recht-Fazel2010}.

When $X$ is a rank $r$ matrix, we have the following inequalities:
\begin{eqnarray}\label{norms ineq}
\|X\| \leq \|X\|_{F} \leq \|X\|_{*}  \leq \sqrt{r} \|X\|_{F} \leq r \|X\|.
\end{eqnarray}

In addition, the support of  a vector $x \in \mathbb{R}^{n}$ is defined as the set containing the indices of the nonzero elements of $x$ \ie, 
\begin{eqnarray}\label{supp}
\supp(x) \defeq \Big\{i \in \{1\ldots n\} \,:\, x_{i} \neq 0\Big\}.
\end{eqnarray}
This definition can be extended to a matrix $X \in \Rbb^{\na\times \nb }$ as follows:
\begin{eqnarray}
\supp(X) \defeq \bigcup_{j} \supp(X_{.,j}),
\end{eqnarray}
which is the union of the supports of all of its columns. We call $X$ a $k$-\emph{joint-sparse} matrix when there are no more than $k$ rows of $X$ containing nonzero elements, $\mathbf{Card}(\supp(X)) \leq k$, where $\mathbf{Card}(.)$ denotes the cardinality of a set. Note that, if $X$ is $k$-joint-sparse the following inequalities hold:
\begin{eqnarray}\label{mix-fro}
 \|X\|_{F} \leq \|X\|_{2,1} \leq \sqrt{k} \|X\|_{F}.
\end{eqnarray}

Finally, throughout this paper we often use the notation $X_{\tcal}$ which is a matrix of the same size as $X$ whose rows indexed by the set $\tcal$ are the same as $X$ and zero elsewhere.

\section{Compressed Sensing}\label{ch3-sec:priorarts}

Compressed sensing \cite{donoho2006compressed, CRT-stable2005} has been introduced in the past decade as an alternative to Shannon's sampling theorem. The original foundations of this new sampling paradigm were built over taking advantage of the \emph{sparsity} of a wide range of natural signals in some properly-chosen orthonormal bases and representing large dimensional sparse vectors by very small number of linear measurements, mainly proportional to their sparsity levels.

\subsubsection{CS Acquisition Protocol}\label{ch3-sec:sampling}

 $m\ll\na\nb$ linear and possibly noisy measurements are collected from the data matrix  $X$ by defining a linear mapping  $\mathcal{A} \, : \, \mathbb{R}^{n_{1} \times n_{2}} \rightarrow \mathbb{R}^{m}$ and its adjoint $\mathcal{A^{*}} \, : \, \mathbb{R}^{m} \rightarrow \mathbb{R}^{n_{1} \times n_{2}}$ as follows:
\begin{eqnarray}\label{ch2-sampling model1}
y = \mathcal{A}(X)+z,
\end{eqnarray}
where $y \in \mathbb{R}^{m}$ is the vector containing the CS measurements and $z \in \mathbb{R}^{m}$ is the noise vector due to quantization or transmission errors. This formulation equivalently has the following explicit matrix representation:
\begin{equation}\label{sampling model2}
y = A X_{vec}+z,
\end{equation}
because any linear mapping $\acal(X)$ can be expressed in a matrix form $\acal(X)\defeq A X_{vec}$, where 
the inner product of the data matrix and each row of $A\in \Rbb^{m \times \na \nb}$ gives a noise-free element of the measurement vector.
In a similar way, we can express the adjoint operator as
\begin{equation*}
\acal^{*}(y) = \textsf{reshape}(A^{T}y,\na,\nb).
\end{equation*}
The $\textsf{reshape}$ operator applies on an $\na\nb$ dimensional vector $A^{T}y$ in order to reshape it into an $\na \times \nb$ matrix such that the first $\na$ elements fill the first column, the next $\na$ elements fill the second column and so on.

The main goal of CS is to robustly recover $X$ from the smallest number of measurements $m$. Despite the simplicity of using a linear sampling scheme and unlike Shannon's linear signal reconstruction, CS decoding applies more complex and nonlinear approximation methods. In the following we mention some of the most relevant approaches to our work:

\subsubsection{Unstructured Sparse Recovery by $\ell_{1}$ Minimization}

The following $\ell_{1}$ problem, also known as Basis Pursuit De-Noising (BPDN), was among the most successful solutions that have
 been initially proposed for CS recovery of a sparse vector $X_{vec}$:
 \begin{eqnarray}\label{ch2-BPDN}
\argmin_{X} \|X_{vec}\|_{\ell_{1}} \qquad \text{subject to}\quad \|y-\acal(X)\|_{\ell_{2}} \leq \varepsilon.
\end{eqnarray}
Note that regularizing the solutions by their
 $\ell_{1}$ norm, although known as the best sparsity-inducing convex functional, can only exploit \emph{unstructured} sparsity in the data.  It has been shown that if the sampling operator $\acal$ satisfies the \emph{restricted isometry property} (RIP), then the solution to  \eqref{ch2-BPDN} can stably recover the underlying data \cite{donoho2006compressed, CRT-stable2005}. For sampling matrices $A$  (or equivalently $\acal$ in \eqref{ch2-sampling model1}) whose elements are drawn independently at random from the Gaussian, Bernoulli or (in general) subgaussian distributions, this property implies the following lower bound on the number of measurements:
\begin{eqnarray}\label{BPDN bound}
m \gtrsim \O \left( k\nb \log(\na/k)\right),
\end{eqnarray}
which indicates a suboptimal performance compared to the limited degrees of freedom of a simultaneous low-rank and joint-sparse data matrix.

%
%
\subsubsection{Joint-Sparse Recovery by $\lm$ Minimization}
To overcome this limitation, some alternative recovery approaches propose to consider a restrictive structure on the sparsity pattern of the data, namely, a \emph{joint-sparsity} structure for the nonzero coefficients of $X$. 
Some approaches based on greedy heuristics \cite{TroppJsparse,duarte2005distributed, DistribuuedTH-SPIE} attempt to identify the joint-sparsity support of $X$, whereas other methods based on convex relaxation \cite{TroppJSconvex,HolgerEldar,Blocksparse} aim at recovering both the joint-sparsity pattern and the nonzero coefficients by penalizing other matrix norms promoting a joint-sparsity structure to their solutions. In particular, penalizing the $\lm$ mixed norm of the data matrix through the following convex problem is very popular and has been widely used in the literature:
 \begin{eqnarray}\label{ch2-l21}
\argmin_{X} \|X\|_{2,1} \qquad \text{subject to}\quad \|y-\acal(X)\|_{\ell_{2}} \leq \varepsilon.
\end{eqnarray}
Despite significant improvements in support recovery\footnote{It has been shown in \cite{DistribuuedTH-SPIE} that a simple one-stage p-Thresholding algorithm can correctly identify the support of an \emph{exact} joint-sparse matrix provided $m \gtrsim \O(k\log(\na))$.}, the joint-sparsity assumption fails to enhance the entire data reconstruction (\ie recovering both the support and the nonzero coefficients). Eldar and Mishali \cite{Blocksparse} extend  the notion of RIP to the \emph{Block}-RIP and show that, if the sampling matrix $A$ satisfies this property, the $\lm$ problem can stably reconstruct joint-sparse matrices. For a dense i.i.d. subgaussian sampling matrix $A$, the number of measurements sufficient for satisfying the Block-RIP is
\begin{eqnarray}\label{l21 bound}
m \gtrsim \O \left( k\log(\na/k) + k\nb \right).
\end{eqnarray}
As we can see, for a large number  $\nb \gg k\log(\na/k)$, this bound indicates only a minor enhancement of a $\log(\na/k)$ factor compared to the measurement bound \eqref{BPDN bound} corresponding to the classical $\ell_{1}$ problem. This is a fundamental limitation arising from  neglecting the underlying correlations among the values of nonzero coefficients \ie, with such a loose assumption (the joint-sparsity model) the data matrix can have $k\nb$ nonzero elements freely taking independent values, and thus generally one requires $\O(k)$ measurements per column for recovery.

\subsubsection{Low-Rank Recovery by the Nuclear Norm Minimization}
Many works have been lately focusing on recovering a low-rank data matrix from an incomplete set of observations. In particular, the \emph{matrix completion} problem \cite{Matirxcompletion, NoisyMC, optspace} deals with recovering a data matrix given only a small fraction of its entries. The key assumption for a stable reconstruction is that the data matrix is highly redundant due to  high linear dependencies among its columns/rows and therefore it has a \emph{low rank}. Similarly, Candes and Plan \cite{candes-plan2009} have shown that if the sampling operator $\acal$ satisfies an extended notion of RIP, called the \emph{Rank}-RIP, then the following \emph{nuclear} norm minimization problem can stably recover the underlying low-rank matrix:
 \begin{eqnarray}\label{ch2-nucmin}
\argmin_{X} \|X\|_{*} \qquad \text{subject to}\quad \|y-\acal(X)\|_{\ell_{2}} \leq \varepsilon.
\end{eqnarray}
In particular, if $\acal$ is drawn from the i.i.d. subgaussian distribution, then the number of measurements sufficient for $\acal$ to satisfy the Rank-RIP (and hence data recovery) is 
\begin{eqnarray}\label{LR bound}
m \gtrsim \O \left( r(\na+ \nb )\right).
\end{eqnarray}
Compared to \eqref{l21 bound} and for a large number $\nb$, we require less number of measurements thanks to properly dealing with the high data correlations via nuclear norm regularization. As we can see in \eqref{LR bound} and compared to the bound \eqref{l21 bound}, the growth of $m$ remains linear  by increasing the number $\nb$ of columns, but scales as $r\nb$ instead of $k\nb$ with  $r \ll k$. For a large number $\nb$, this bound indicates that  $\O(r)$ measurements per column are sufficient which is proportional to the number of principal components of $X$ rather than the sparsity level of each column $\O(k)$ in \eqref{l21 bound} or \eqref{BPDN bound}. Despite this improvement, for $\na \gg \nb$ ignoring sparsity would significantly decrease the efficiency of this approach compared to the few degrees of freedom of a low-rank and joint-sparse data matrix.

\subsubsection{Rank Aware Joint-Sparse Recovery}\label{sec: prior-RAJS}
Considering those previously described approaches and their weak and strong points, one can think of the idea of combining both data priors together and reconstruct matrices that are \emph{simultaneously} low-rank and joint-sparse. Recently a few works \cite{Davies-Eldar2010, Korean-Music2010} have considered \emph{rank awareness} in data recovery from multiple measurement vectors (MMV). More precisely, the joint-sparse MMV inverse problem \cite{MMV-Rao, MMV-Chen, eldar2009compressed, TroppJsparse}
focuses on recovering a joint-sparse matrix $X$ from a set of CS measurements $Y \in \mathbb{R}^{\mhat \times n_{2}}$ acquired by $Y = \atilde X$. We refer to such acquisition model as the \emph{uniform} sampling because a unique sampling matrix $\atilde\in\Rbb^{m/\nb \times \na}$ applies on all columns of $X$. 

Davies and Eldar \cite{Davies-Eldar2010} proposed a specific rank-aware \emph{greedy} algorithm, that in case of using a random i.i.d. Gaussian sampling matrix $\atilde$ is able to recover (with a high probability) an \emph{exact} $k$-joint-sparse and rank $r$ $X$ from the noiseless MMV provided
\begin{eqnarray}\label{davis bound}
m = n_{2} \mhat \gtrsim \mathcal{O} \Big( n_{2} k \left( 1+\frac{ \log{n_{1}}}{r} \right)\Big).
\end{eqnarray}
It can be viewed from \eqref{davis bound} that this approach is more suitable for applications with \emph{high-rank} data matrices, where the columns of $X$ are as independent as possible from each other. This is contrary to our assumption of having a redundant, low-rank data structure. This limitation comes from the fact that
 \emph{the uniform sampling scheme is not generally suitable for acquisition of correlated data matrices}.
Given highly redundant data, such sampling mechanism results in a redundant set of measurements ($X$ is low-rank and so becomes $Y$, because $\rank(Y) \leq \rank(X)$). 
 Think of the  extreme example of having a joint-sparse rank one matrix whose columns are identical to each other. Using a uniform sampling scheme would result in multiple copies of the same measurements for all the columns and thus, 
applying any joint decoding scheme exploiting the intra-column correlations would be pointless as the measurements do not carry any new information across the columns. 


%
\section{Simultaneous Low-rank and Joint-Sparse Matrix Approximation}\label{ch3-sec:rec-approach}


Data reconstruction from CS measurements typically reduces to solving an underdetermined and noisy linear system of equations such as \eqref{ch2-sampling model1} where $m\ll\na\nb$. Without any further assumptions \eqref{ch2-sampling model1} would admit infinitely many solutions and thus, recovery of the original data matrix is not generally feasible. 
As previously discussed, however, data matrices of interest in this paper have a particular simultaneous low-rank and joint-sparse structure restricting data to have very few degrees of freedom, \ie less than both the matrix dimension and the dimension of the underlying sparse coefficients. 
Our main goal in this part is to propose a robust recovery scheme which incorporates properly these constrains in order to make \eqref{ch2-sampling model1} invertible with as few measurements as possible. 


Recovery of a joint-sparse and rank $r$ data matrix from a set of noisy CS measurements can be casted as the solution to the following minimization problem:
\begin{align}\label{P0}
\argmin_{X} \quad &   \mathbf{Card}(\supp(X))  \\
\text{subject to}\quad  & \left\| y - \mathcal{A}(X) \right\|_{\ell_{2}} \leq \varepsilon,  \nonumber \\
& \rank(X) \leq r,  \nonumber 
\end{align}
where $\varepsilon$ is an upper bound on the energy of the noise, $\|z\|_{\ell_2} \leq \varepsilon$. 
Note that finding the sparsest solution to a linear system of equations, in general, is an NP-hard problem \cite{natarajan1995sas}. Problems such as minimizing $\mathbf{Card}(\supp(X))$ or the rank minimization are similarly reducing to the cardinality minimization problem of finding the sparsest element in an affine space and thus in general, finding solution to a problem like \eqref{P0} is computationally intractable. The original solution must be approximated by a polynomial-time method. 

A trivial solution could be the \emph{cascade} of two classical CS recovery schemes:
One identifies the joint-support set and the other evaluates the rows of the matrix corresponding to the recovered support using a low-rank approximation scheme. Unfortunately, the application of this approach is mainly limited to the recovery of the \emph{exact} low-rank and joint-sparse matrices. 
In practice signals are almost never exactly sparse, 
and using the cascade strategy (\eg recovering first the dominant $k$-joint-sparse part of data, and then extracting a rank $r$ matrix from it) for dense data matrices might not be always optimal in order to identify the most dominant component  having simultaneously both structures.



In order to provide \emph{robustness} against non-exact but approximately low-rank and joint-sparse data, we  incorporate both data priors into a single recovery formulation.
We relax \eqref{P0} by replacing the non-convex terms $\rank(X)$ and $\mathbf{Card}(\supp(X))$ by the convex terms inducing similar structures to their solutions \ie, the nuclear norm $\|X\|_{*}$ and the $\|X\|_{2,1}$ mixed-norm, correspondingly. We cast the CS recovery problem as solving one of the following three convex problems:

\begin{enumerate}
\item Having data nuclear norm as the regularization parameter, or provided with an estimation of it, we solve:
\begin{align}\label{P1-1}
\argmin_{X} \quad &   \|X\|_{2,1}  \\
\text{subject to}\quad  & \left\| y - \mathcal{A}(X) \right\|_{\ell_2} \leq \varepsilon , \nonumber \\
& \|X\|_{*} \leq \tau.  \nonumber 
\end{align}

\item Having the $\lm$ mixed-norm of data as the regularization parameter, or provided with an estimation if it, we solve:
\begin{align}\label{P1-2}
\argmin_{X} \quad &   \|X\|_{*}  \\
\text{subject to}\quad  & \left\| y - \mathcal{A}(X) \right\|_{\ell_2} \leq \varepsilon , \nonumber \\
& \|X\|_{2,1} \leq \gamma.  \nonumber 
\end{align}

\item Having a regularization factor $\lambda$ balancing between the $\lm$ norm and the nuclear norm, we solve:
\begin{align}\label{P1-3}
\argmin_{X} \quad & \|X\|_{2,1} + \lambda \|X\|_{*}  \\
\text{subject to}\quad  & \left\| y - \mathcal{A}( X) \right\|_{\ell_2} \leq \varepsilon. \nonumber 
\end{align}
\end{enumerate}


We explain in Section \ref{sec: PPXA} that all three problems above can be solved iteratively (in polynomial time) using  the parallel proximal splitting algorithm (PPXA) \cite{proximal-splitting}. 
Note that all three formulations  \eqref{P1-1},  \eqref{P1-2} and \eqref{P1-3}  contain a regularization parameter (either $\tau$, $\gamma$ or $\lambda$) that should be  correctly tuned. In the next section  we show  that, for properly-chosen parameters, the solutions of these convex problems can robustly  approximate the original data  provided by very few linear measurements

\section{Theoretical Guarantees}\label{ch3-sec:theory}

In this section we provide the theoretical analysis of the recovery problems \eqref{P1-1}-\eqref{P1-3}. We prove that for certain linear mappings $\acal$, these three convex relaxations can exactly reconstruct all simultaneous low-rank and joint-sparse data from their corresponding CS measurements. In addition, our results guarantee the stability of our approach against noisy measurements and non-exact low-rank and joint-sparse data. 
For a certain class of random linear mappings our results establish a  lower bound on the number $m$ of CS measurements in terms of  the joint-sparsity level $k$ and the rank $r$ of the underlying data matrix.



\subsection{Sampling Operator RIP}
We define a specific restricted isometry property (RIP) that later helps us to build our main theoretical results. 
\begin{definition}\label{def: RIP}
For positive integers $k,r\in \Nbb_{+} $, a linear map $\mathcal{A}$ satisfies the restricted isometry property, if for all simultaneous rank $r$ and $k$-joint-sparse matrices $X$ we have the following inequalities:
\begin{eqnarray}\label{RIP}
 (1- \delta_{r,k}) \|X\|_{F}^{2} \leq  \|\mathcal{A}(X)\|^{2}_{\ell_2} \leq (1+ \delta_{r,k}) \|X\|_{F}^{2}.
\end{eqnarray}
The RIP constant $\delta_{r,k}$ is the smallest constant for which the property above holds.
\end{definition}
As we can see and compared to  the two former definitions of the Block-RIP \cite{Blocksparse} and the Rank-RIP \cite{candes-plan2009, Recht-Fazel2010},
Definition \ref{def: RIP} provides a more restrictive property which holds for the intersection set of the low-rank and the joint (block) sparse matrices.
Verifying whether a linear map satisfies RIP is generally combinatorial and thus prohibitive for large size problems. However, for linear mappings drawn at random from a certain class of distributions and thanks to their well-established concentration properties, the following theorem and corollary guarantee RIP: 

\begin{theorem}\label{RIPth}
Let $\mathcal{A}: \Rbb^{\na\times \nb}\rightarrow \Rbb^{m}$ be a random linear map  obeying the following concentration bound for any $X \in \mathbb{R}^{n_{1} \times n_{2}}$ and $0<t <1$,
\begin{eqnarray}\label{concentration}
\mathbf{Pr}  \left(  \,  \left|  \| \mathcal{A}(X) \|_{\ell_2}^{2} - \|X\|^{2}_{F} \right| > t \|X\|_{F}^{2} \,  \right) \leq C \exp \left( -c \, m\right),
\end{eqnarray}
where $C$ and $c$ are fixed constants given $t$. Then $\mathcal {A}$ satisfies RIP with constant $\delta_{r,k}$ with probability greater than $1- C e^{- \kappa_{0} m}$, if the number of measurements is greater than 
\begin{eqnarray}\label{sufficientmeas}
m \geq \kappa_{1} \Big( k  \log(n_{1}/k) + (k+1)r   + n_{2} r \Big),
\end{eqnarray}
where $\kappa_{0}$ and $\kappa_{1}$ are fixed constants for a given $\delta_{r,k}$.
\end{theorem}

\begin{corollary}\label{corollary1}
For a linear map $\mathcal{A}: \Rbb^{\na\times \nb}\rightarrow \Rbb^{m}$ corresponding to a matrix $A\in \Rbb^{m\times \na\nb}$ whose elements are drawn independently at random from the Gaussian $\mathcal{N}(0,1/m)$, Bernoulli $\{ \pm1/ \sqrt{m}\}$ or subgaussian (with a proper normalization) distributions, RIP holds whenever the number of measurements satisfies \eqref{sufficientmeas}. 
\end{corollary}

The proof of Theorem \ref{RIPth} is provided in Section \ref{sec:ch2-proofRIP} and Corollary \ref{corollary1} is a straightforward application of Theorem \ref{RIPth}. Indeed, for the Gaussian distribution since $\| \mathcal{A}(X) \|^{2}_{\ell_2}$ is a chi-squared random variable with $m$ degrees of freedom multiplied by $\|X\|^{2}_{F}/m$, the following concentration bound holds:
\if@twocolumn
\begin{align}
&\mathbf{Pr}  \left(  \,  \left|  \| \mathcal{A}(X) \|_{\ell_2}^{2} - \|X\|^{2}_{F} \right| > t \|X\|_{F}^{2} \,  \right) \nonumber\\
&\quad \leq 2 \exp \Big( - \frac{m}{2} (t^{2}/2-t^{3}/3) \Big).
\end{align}
\else
\begin{align}
\mathbf{Pr}  \left(  \,  \left|  \| \mathcal{A}(X) \|_{\ell_2}^{2} - \|X\|^{2}_{F} \right| > t \|X\|_{F}^{2} \,  \right) 
\leq 2 \exp \Big( - \frac{m}{2} (t^{2}/2-t^{3}/3) \Big).
\end{align}
\fi
Similar bounds hold for the Bernoulli and subgaussian ensembles with different constants. Regarding Theorem \ref{RIPth}, random linear maps drawn from these distributions indeed satisfy the deviation bound \eqref{concentration} and thus RIP holds for them provided \eqref{sufficientmeas}.

\subsection{Main Reconstruction Error Bounds}
 Our main performance bounds here are derived based on the new notion of RIP defined above.
Given a set of linear measurements that are collected from a matrix $X\in\Rbb^{\na \times \nb}$ (which is not necessarily low-rank or joint-sparse) through the sampling model \eqref{ch2-sampling model1} and by using a sampling operator $\acal$ satisfying RIP, the following theorems then bound the approximation error of the recovery problems \eqref{P1-1}-\eqref{P1-3}:







\begin{theorem}\label{main theorem}
Suppose $\delta_{6r,2k} < 1/7$ and $\|z\|_{\ell_2} \leq \varepsilon$. Let $X_{r,k}$ be any rank $r$ and k-joint-sparse matrix and $\tcal_{0}=\supp(X_{r,k})$. 

\begin{enumerate} [(a)]
\item For  $\tau = \|X\|_{*}$, the solution $\widehat X$ to \eqref{P1-1} obeys the following bound: 
\begin{align}
\hspace{-.5cm} \|\widehat X - X\|_{F} \leq \kappa'_{0} \left(     \frac { \| X - X_{\tcal_{0}} \|_{2,1}} {\sqrt{k}} + \frac{\| X - X_{r,k} \|_{*}}{\sqrt{2r}}   \right) + \kappa'_{1} \varepsilon.\label{error bound}
\end{align}
Constants $\kappa'_{0}$ and $\kappa'_{1}$ are fixed for a given $\delta_{r,k}$.  

\item For $\gamma = \|X\|_{2,1}$, the solution $\widehat X$ to \eqref{P1-2} obeys the bound \eqref{error bound}.

\item For $\lambda = \sqrt{\frac{k}{2r}}$, the solution $\widehat X$ to \eqref{P1-3} obeys the bound \eqref{error bound}.

\end{enumerate}
\end{theorem}




\begin{corollary} \label{main corollary}
Suppose $\delta_{6r,2k} < 1/7$ and the measurements are noiseless. Then, all exact (simultaneous) $k$-joint-sparse and rank $r$ matrices $X$ can be exactly recovered by \eqref{P1-1}, \eqref{P1-2} or \eqref{P1-3} (with proper regularization parameters) \ie, $\widehat X=X$.
\end{corollary}

The proof of Theorem \ref{main theorem} is provided in Section \ref{ch2-mainproofs}.
The proof of the corollary is trivial by setting $X_{r,k}=X_{\tcal_{0}}=X$ and $\varepsilon = 0$ in \eqref{error bound}. Given noisy measurements and a data matrix $X$ that is not exactly low-rank or joint-sparse, the reconstruction error \eqref{error bound} is bounded by two terms. First, it is proportional (with a constant factor) to the power of the sampling noise. Second, it is determined by the best rank $r$ and $k$-joint-sparse approximation of $X$ that minimizes the first term of \eqref{error bound}. This term can be interpreted as a measure of data \emph{compressibility} \ie, how accurately  data can be approximated by a simultaneous low-rank and joint-sparse matrix.

According to Corollary \ref{corollary1}, for sampling matrices $A$ drawn independently at random from the Gaussian, Bernoulli or subgaussian distributions (and more generally, all random linear maps $\acal$ satisfying \eqref{concentration}), Theorem \ref{main theorem} guarantees that our proposed reconstruction approach based on joint nuclear and $\lm$ norms minimization can stably recover all simultaneous rank $r$ and $k$-joint-sparse matrices $X$ (with an overwhelming probability) provided 
\begin{eqnarray}\label{ch2-nuc-l21}
m \gtrsim \mathcal{O} \big( k  \log(n_{1}/k) + k r  + n_{2} r \big),\nonumber
\end{eqnarray}
which, as previously mentioned, indicates a near-optimal recovery performance regarding the degrees of freedom of those data matrices \ie, $r(k+\nb-r)$. In addition, this bound reveals that our proposed recovery approach requires significantly less number of CS measurements to reconstruct such structured data compared to the prior art methods (see \eqref{l21 bound}, \eqref{LR bound} and \eqref{davis bound}). Unlike the $\lm$ approach \eqref{ch2-l21} and for a large number  $\nb$ of columns, our approach requires $\O(r)$ measurements per column, which is proportional to the number of principal components of this data. In addition, compared to the nuclear norm minimization scheme \eqref{ch2-nucmin}, we achieve a shaper  lower bound on $m$ that is proportional to the joint-sparsity level $k$ of the data matrix and scales sub-linearly with the column dimensionality $\na$. Finally, our bound is of a different nature than (\ref{davis bound}): the lower the rank, less measurements are required. This reflects the importance of a good design for the sampling scheme $\mathcal{A}$ along with a recovery approach benefiting efficiently from data structures.

%

Note that Theorem \ref{main theorem} precisely characterizes  the proper values for the regularization parameters in \eqref{P1-1}-\eqref{P1-3}. If there is no prior knowledge available about the nuclear norm or the $\lm$ norm of the underlying data (and a blind tuning might be challenging), \eqref{P1-3} brings more flexibility to the reconstruction process. For a given $\lambda$, Theorem \ref{main theorem}-c guarantees that the recovery error is bounded by the minimum of \eqref{error bound} for the best pair $r, k$ satisfying both RIP (\ie$\delta_{6r,2k} < 1/7$) and  $\lambda= \sqrt{k/2r}$.



\if@twocolumn
\else
\fi
\section{Proofs}\label{sec:proofs}

\subsection{Proof of Theorem \ref{RIPth}} \label{sec:ch2-proofRIP}

Let $\widetilde X_{\tcal}$ denote the $k \times n_{2}$ submatrix formed by selecting the rows of $X$ indexed by a set $\mathcal{T} \subset \{1,\ldots, \na\}$ with the cardinality $k$. Correspondingly, for a given $\mathcal{T}$ we define the restricted map $\acal_{\tcal} : \mathbb{R}^{k\times n_{2}} \rightarrow \mathbb{R}^m$ which applies on $\widetilde X_{\tcal}$, and  is related to $\acal$ through its explicit matrix formulation \eqref{sampling model2} as follows:
\begin{eqnarray}
 \acal_{\tcal}(\widetilde X_{\tcal})  \defeq \widetilde A_{\tcal}\,(\widetilde X_{\tcal})_{vec}, \nonumber
\end{eqnarray}
where, $\widetilde A_{\tcal} \in \mathbb{R}^{m \times k\nb}$ is the submatrix containing the columns of $A$ indexed by $\tcal$. If  we assume $X$ is an exact $k$-joint-sparse matrix whose rows are supported over the set $\mathcal{T}$ and zero elsewhere, we will have 
\[
\acal(X) = \mathcal{A}_{\mathcal{T}}(\widetilde X_{\tcal}).
\]
In this case and according to our assumption (\ie, $\acal$ satisfies \eqref{concentration}) we can write  a similar concentration bound for a given $\tcal$ and $\widetilde X_{\tcal}$:
\begin{align}\label{concentration2}
\hspace{-.1cm}\mathbf{Pr}  \left(  \,  \left|  \| \mathcal{A}_{\mathcal{T}} (\widetilde X_{\mathcal{T}}) \|_{\ell_2}^{2} - \|\widetilde X_{\mathcal{T}}\|^{2}_{F} \right| > t \|\widetilde X_{\mathcal{T}}\|_{F}^{2} \,  \right) \leq C \exp \left( -c \, m\right).
\end{align}
According to the Rank-RIP property defined in \cite{candes-plan2009}, if a linear map $\mathcal{A}_{\mathcal{T}}$ satisfies the bound \eqref{concentration2} then, for a given $\tcal$ and \emph{all} rank $r$ realizations of $\widetilde X_{\mathcal{T}}$, we have
\if@twocolumn
\begin{align}\label{rankRIP}
&\mathbf{Pr} \left( \left| \| \mathcal{A}_{\mathcal{T}} (\widetilde X_{\mathcal{T}} )\|^{2}_{\ell_2} - \|\widetilde X_{\mathcal{T}}\|_{F}^{2} \right|  >  \delta \|\widetilde X_{\mathcal{T}}\|_{F}^{2} \right) \nonumber \\
&\quad \leq C \exp{ \left( (k+n_{2}+1)r \log(36\sqrt 2/\delta)-c \, m \right)},
\end{align}
\else
\begin{align}\label{rankRIP}
\mathbf{Pr} \left( \left| \| \mathcal{A}_{\mathcal{T}} (\widetilde X_{\mathcal{T}} )\|^{2}_{\ell_2} - \|\widetilde X_{\mathcal{T}}\|_{F}^{2} \right|  >  \delta \|\widetilde X_{\mathcal{T}}\|_{F}^{2} \right) 
 \leq C \exp{ \left( (k+n_{2}+1)r \log(36\sqrt 2/\delta)-c \, m \right)},
\end{align}
\fi
where $0<\delta <1$ is a fixed constant. The proof of this statement mainly uses a \emph{covering number} argument for low-rank matrices and it can be found in \cite{candes-plan2009} (see theorem 2.3).
Now, by considering \emph{all} possible combinations for the subset $\mathcal{T}$ \ie, ${{n_{1}}\choose{k}} \leq \left( \frac{n_{1}}{k} e \right)^{k} $ and applying a union bound on \eqref{rankRIP}, we can deduce \eqref{RIP} for \emph{all} $k$-joint-sparse and rank $r$ matrices $X$ with probability greater than
\if@twocolumn
\begin{align}
1-C \exp \Big(& k+k\log(\na/k)  \nonumber\\
&+(k+\nb+1)r \log(36\sqrt 2/\delta)-c \, m \Big). \nonumber
\end{align}
\else
\begin{align}
1-C \exp \Big(& k+k\log(\na/k)  
+(k+\nb+1)r \log(36\sqrt 2/\delta)-c \, m \Big). \nonumber
\end{align}
\fi
Equivalently, with probability exceeding $1-C e^{-\kappa_{0} m}$, $\acal$ satisfies RIP with constant $\delta_{r,k} \leq \delta$,  whenever
\begin{eqnarray}
m \geq \kappa_{1} \Big( k  \log(n_{1}/k) + (k   + n_{2}+1) r \Big), \nonumber
\end{eqnarray}
where $\kappa_{0}$ and $\kappa_{1}$ are fixed constants for a given $\delta$ \ie, $\kappa_{0} = c - \log(36\sqrt 2/\delta) /\kappa_{1}$, and $\kappa_{1} > \log(36\sqrt 2/\delta)/ c $.

\subsection{Proof of Theorem \ref{main theorem}}\label{ch2-mainproofs}

In this part we first provide the proof of Theorem \ref{main theorem}-a and later, we use similar techniques to prove Theorems \ref{main theorem}-b and \ref{main theorem}-c.
For this purpose, we prove the following two lemmas that are the key elements for the proof of  Theorem \ref{main theorem}.

An important implication of the RIP is that for linear maps $\acal$ having small RIP constants, the measurement outcomes of any two $k$-joint-sparse and rank $r$ orthogonal matrices become \emph{nearly} orthogonal. This statement is precisely characterized by the following lemma:

\begin{lemma}\label{RIPlemma}
$\forall X, X'$ such that, $\rank(X) \leq r$, $\rank(X') \leq r'$, $\mathbf{Card}(\supp(X)) \leq k$, $\mathbf{Card}(\supp(X'))\leq k'$ and  $\langle X,X' \rangle = 0$, we have:
\begin{equation}\label{innerprod}
\big| \langle \mathcal{A}(X), \mathcal{A}(X') \rangle \big| \leq \delta_{r+r',k+k'} \|X\|_{F} \, \|X'\|_{F}
\end{equation}
\end{lemma}

\begin{proof}
Without loss of generality assume $X$ and $X'$ have unit Frobenius norms. Since $X \pm X'$ is at most $k+k'$ joint-sparse and $\rank(X \pm X') \leq r+r'$, we can write
\if@twocolumn
\begin{align}
(1- \delta_{r+r',k+k'}) \|X\pm X'\|_{F}^{2} \leq&\,\,  \|\mathcal{A}(X\pm X')\|^{2}_{\ell_2} \nonumber\\
\leq& \,\,(1+ \delta_{r+r',k+k'}) \|X\pm X'\|_{F}^{2}.\nonumber
\end{align}
\else
\begin{align}
(1- \delta_{r+r',k+k'}) \|X\pm X'\|_{F}^{2} \leq  \|\mathcal{A}(X\pm X')\|^{2}_{\ell_2} 
\leq(1+ \delta_{r+r',k+k'}) \|X\pm X'\|_{F}^{2}.\nonumber
\end{align}
\fi
Moreover, $X$ and $X'$ are orthogonal to each other which implies that $\|X\pm X'\|_{F}^{2} = \|X\|_{F}^{2}+\|X'\|_{F}^{2}=2$. Now, by using the parallelogram identity we can deduce:
\if@twocolumn
\begin{align}
\big| \langle \mathcal{A}(X), \mathcal{A}(X') \rangle \big| =&\sp \frac{1}{4}  \Big| \|\mathcal{A}(X + X')\|_{\ell_{2}}^{2} - \| \mathcal{A}(X - X')\|_{\ell_{2}}^{2} \Big| \nonumber \\
\leq&\sp \delta_{r+r',k+k'},\nonumber
\end{align}
\else
\begin{align}
\big| \langle \mathcal{A}(X), \mathcal{A}(X') \rangle \big| = \sp \frac{1}{4}  \Big| \|\mathcal{A}(X + X')\|_{\ell_{2}}^{2} - \| \mathcal{A}(X - X')\|_{\ell_{2}}^{2} \Big| 
\leq \delta_{r+r',k+k'},\nonumber
\end{align}
\fi
which completes the proof.
\end{proof}

\begin{lemma}\label{decomposition}
Let $A$ and $B$ be matrices of the same size and $\rank(A)\leq r$, $\supp(A)=\mathcal{T}$. $\exists \, B_{1}, B_{2}$ (matrices of the same size as $A$ and $B$) such that:
\begin{enumerate}
\item $B= B_{1}+B_{2}$,
\item  $\supp (B_{1})=\mathcal{T}$, $\rank(B_{1})\leq 2r$,
\item $A^{T}B_{2}=\mathbf{0}$, $A(B_{2})^{T}= \mathbf{0}$,
\item $\langle B_{1}, B_{2}\rangle = 0$.
\end{enumerate}
 
\end{lemma}

\begin{proof}
Suppose $\widetilde A_{\tcal}$ and $\widetilde B_{\tcal}$ are submatrices correspondingly containing the rows of $A$ and $B$ indexed by $\tcal$. Recht \etal,  (see Lemma 3.4 in \cite{Recht-Fazel2010}) propose a decomposition $\widetilde B_{\tcal} = \widetilde B_{1}+\widetilde B_{2}$ so that 
\begin{itemize}
\item $\rank(\widetilde B_{1}) \leq 2 \rank(\widetilde A) \leq 2r$,
\item $(\widetilde A_{\tcal})^{T} \widetilde B_{2}=\mathbf{0}$, $\widetilde A_{\tcal} (\widetilde B_{2})^{T}=\mathbf{0}$,
\item $\langle \widetilde B_{1}, \widetilde B_{2}\rangle=0$.
\end{itemize}

On the other hand, lets decompose $B = B_{\mathcal{T}} + B_{\mathcal{T}^{c}}$ where $B_{\mathcal{T}}$ and  $B_{\mathcal{T}^{c}}$ are respectively the same as $B$ on the index sets $\mathcal{T}$ and $\mathcal{T}^{c}$ (the set complement of $\mathcal{T}$) and zero elsewhere. Suppose $B_{1_{\tcal}}$ and $B_{2_{\tcal}}$ are matrices of the same size as $B$, whose rows indexed by $\tcal$ are respectively populated by the submatrices $\widetilde B_{1}$ and $\widetilde B_{2}$, and the remaining elements are zero. We can write $B_{\mathcal{T}}  = B_{1_{\tcal}} + B_{2_{\tcal}}  $. Now, by letting $B_{1} = B_{1_{\tcal}}$ and $B_{2} = B_{2_{\tcal}} + B_{\mathcal{T}^{c}}$ one can easily verifies that $B_{1}$ and $B_{2}$ satisfy properties (1)-(4).

\end{proof}

\subsubsection{Decomposition of The Error Term}
Here, by using Lemma \ref{decomposition} we decompose the error term $H=X-\widehat X$ into the summation of rank $2r$ and $k$-joint-sparse matrices that are mutually orthogonal to each other. This decomposition later appears to be essential to bound the error using Lemma \ref{RIPlemma}.

Let us fix the notation so that $X_{\tcal}\in\Rbb^{\na\times \nb}$ is a joint-sparse matrix which is identical to $X$ for the rows indexed by $\tcal$, and zero elsewhere. Let $\tcal_{0}$ denote the joint support of $X_{r,k}$ \ie, $\tcal_{0} =  \supp(X_{r,k})$ with $\card(\tcal_{0}) \leq k$. Since $X_{r,k}$ is additionally a rank $r$ matrix, we can use Lemma \ref{decomposition} to decompose $H$ into orthogonal matrices $H = H^{0}_{\tcal_{0}} + \Hbar$ such that, 
\begin{itemize}
\item $\rank(H^{0}_{\tcal_{0}}) \leq 2r$, 
\item $\langle H^{0}_{\tcal_{0}} , \Hbar \rangle = 0$,
\item $(X_{r,k})^{T} \Hbar =X_{r,k} \Hbar^{T}=\mathbf{0}$.
\end{itemize}

Further, we decompose $\Hbar$ into the summation of disjoint $k$-joint-sparse matrices. For this purpose we write $\Hbar = \Hbar_{\tcal_{0}}+ \Hbar_{\tcal_{0}^{c}}$, where $\tcal_{0}^{c}$ denotes the set complement of the support $\tcal_{0}$. Now we decompose $\Hbar_{\tcal_{0}^{c}}$ into the following terms: 
\[ 
\Hbar_{\tcal_{0}^{c}} = \sum_{i\geq 1} \Hbar_{\tcal_{i}},
\]
 where $\tcal_{1}$ is the set containing the indices of the $k$ rows of $\Hbar_{\tcal_0^{c}}$ with the largest $\ell_{2}$ norm, $\tcal_{2}$ contains the indices of the next $k$ high energy rows and so on. By construction $\Hbar_{\tcal_{i}}$s are supported on disjoint sets. We simplify the notations by using $\Hbar_{i}$ instead of $\Hbar_{\tcal_{i}}$. 

Next, we decompose each $\Hbar_{i}$  into the summation of rank $2r$ matrices: 
\[\Hbar_{i} = \sum_{j\geq 1} \Hbar_{i,j}.\]
We use here a similar decomposition proposed in \cite{candes-plan2009}: let $\Hbar_{i} = U \Sigma V^{T}$ be the ordered singular value decomposition of $\Hbar_{i}$ \ie, diag($\Sigma$) is the ordered sequence of the singular values ($\Sigma_{i,i} \geq \Sigma_{i+1,i+1} $). For integers $j \geq  1$, define the index sets $J_{j} := \{ 2r(j-1)+1, ..., 2rj\}$. Then, the decomposition terms will be 
\[
\Hbar_{i,j} := U_{J_{j}}  \Sigma_{J_{j}} (V_{J_{j}})^{T}, \]
\ie $\Hbar_{i,1}$ corresponds to the first $2r$ largest singular values of $\Hbar_{i}$,  $\Hbar_{i,2}$ corresponds to the next $2r$ largest singular values of $\Hbar_{i}$ and so on.

Those two steps suggest a decomposition  $\Hbar = \sum_{i \geq 0,j \geq 1} \Hbar_{i,j}$, whose terms $\Hbar_{i,j}$ are $k$-joint-sparse and rank $2r$ matrices mutually orthogonal to each other \ie, $\forall (i,j)\neq (i',j')$ 
\[
\langle \Hbar_{i,j}, \Hbar_{i',j'}\rangle =0.
\]
Additionally, $ H^{0}_{\tcal_{0}} $ is orthogonal to all $\Hbar_{i,j}$ as a result of either having disjoint supports (\ie, for $ i\geq 1$), or by the specific constructions of  $H^{0}_{\tcal_{0}}$ (see Lemma \ref{decomposition} in this paper, and Lemma 3.4 in \cite{Recht-Fazel2010}) and $\Hbar_{0,j}$.

Standard computations show that the following inequality holds for all $i$:
\begin{eqnarray} \label{lF-l*}
\sum_{j\geq 2} \| \Hbar_{i,j} \|_{F} \leq \frac{  \| \Hbar_{i} \|_{*}  } {\sqrt {2r} }.
\end{eqnarray}
Also, we have
\begin{eqnarray}\label{lF-l21}
\sum_{i\geq 2} \| \Hbar_{i} \|_{F} \leq \frac{ \| \Hbar_{\tcal_{0}^{c}} \|_{2,1}} {\sqrt k}.
\end{eqnarray}
Inequalities \eqref{lF-l*} and \eqref{lF-l21} are respectively the consequences of 
\[
 \| \Hbar_{i,j}\|_{F}  \leq (2r)^{-1/2}  \| \Hbar_{i,j-1}\|_{*}
 \]  
 and for $i>1$ 
 \[ \| \Hbar_{i} \|_{F} \leq k^{-1/2} \| \Hbar_{i-1} \|_{2,1}.\] 
As a result, the following inequalities hold:
\begin{align}
\mathcal{S} \defeq& \sum_{ \substack{ i\geq 0, j\geq 1\\ (i,j) \neq (1,1),(0,1)}} \| \Hbar_{i,j} \|_{F} \nonumber\\
 \leq&\,\, \sum_{i \geq 2} \| \Hbar_{i} \|_{F} + \sum_{i\geq 0, j \geq 2} \| \Hbar_{i,j} \|_{F} \\
\leq&\,\, \frac{1}{\sqrt k} \|  \Hbar_{\tcal_{0}^{c}}\|_{2,1}  + \frac{1}{\sqrt{2r}} \sum_{i\geq0} \| \Hbar_{i} \|_{*} \\
\leq&\,\, \frac{1}{\sqrt k} \| \Hbar_{\tcal_{0}^{c}} \|_{2,1}  + \frac{1}{\sqrt{2r}}  \| \Hbar \|_{*}. \label{sumFro}
\end{align}
\newline

\subsubsection{Proof of the Main Error Bounds}
We start with the proof of Theorem \ref{main theorem}-a. Since $\acal (X)$ is corrupted by a noise whose $\ell_{2}$ norm is assumed to be bounded by $\varepsilon$, and $\widehat X$ is a feasible point for \eqref{P1-1} we can write,
\begin{align}
\|\acal(H)\|_{\ell_2} =&\,\, \|\acal(X-\widehat X)\|_{\ell_2} \nonumber \\ 
\leq&\,\, \|y- \acal(X)\|_{\ell_2}+\|y-\acal(\widehat X)\|_{\ell_2} \leq 2 \varepsilon.\label{epsilon}
\end{align}

By assumption $\|X\|_{*} = \tau$, and since $\widehat X$ is a feasible point of \eqref{P1-1} satisfying $\|\widehat X\|_{*} \leq \tau$ we have
\begin{align}
\| X \|_{*} \geq &\,\,  \|X + H\|_{*} \nonumber \\
\geq &\,\, \|X_{r,k} + \Hbar \|_{*}  - \| \overline X_{r,k}\|_{*} - \|H^{0}_{\tcal_{0}}\|_{*} \nonumber \\
=&\sp \| X_{r,k} \|_{*} + \| \Hbar  \|_{*}   - \| \overline X_{r,k} \|_{*} - \| H^{0}_{\tcal_{0}} \|_{*}, \label{nuc separation}
\end{align}
which yields the following inequality:
\begin{eqnarray}\label{H*}
\|\Hbar \|_{*} \leq \|H^{0}_{\tcal_{0}}\|_{*} + 2 \|\overline X_{r,k}\|_{*}.
\end {eqnarray}
Note that equality \eqref{nuc separation} holds because both row and column spaces of $X_{r,k}$ and $\Hbar$ are orthogonal to each other. On the other hand, as a result of the minimization \eqref{P1-1} we have
\begin{align}
\|X\|_{2,1} \geq&\sp  \| X+H \|_{2,1} \nonumber \\
=&\sp \| X_{\tcal_{0}} + H_{\tcal_{0}} \|_{2,1} + \| X_{\tcal^{c}_{0}} + H_{\tcal^{c}_{0}} \|_{2,1} \nonumber \\
\geq&\sp \| X_{\tcal_{0}}\|_{2,1} - \|H_{\tcal_{0}} \|_{2,1} - \| X_{\tcal^{c}_{0}} \|_{2,1} + \| H_{\tcal^{c}_{0}} \|_{2,1}. \nonumber
\end{align}
Consequently, we can deduce:
\begin{eqnarray} \label{H21}
\| H_{\tcal_{0}^{c}} \|_{2,1} \leq \| H_{\tcal_{0}} \|_{2,1} + 2\| X_{\tcal_{0}^{c}} \|_{2,1}.
\end{eqnarray}
Note that $H_{\tcal_{0}^{c}} =\Hbar_{\tcal_{0}^{c}}$. Now, dividing both sides of \eqref{H21} by $\sqrt k$ and using \eqref{mix-fro} yields: 
\begin{align}
\frac{1} {\sqrt k}  \| \Hbar_{\tcal_{0}^{c}}  \|_{2,1}  \leq&\sp \| H_{\tcal_{0}} \|_{F}  +  2e  \nonumber \\
\leq&\sp  \| H^{0}_{\tcal_{0}} \|_{F} + \| \Hbar_{0,1}\|_{F} +  \sum_{j\geq2} \| \Hbar_{0,j}\|_{F}  +2e \nonumber \\
\leq&\sp  \| H^{0}_{\tcal_{0}} \|_{F} + \| \Hbar_{0,1}\|_{F} +   \frac{1}{\sqrt{2r}} \| \Hbar \|_{*}  +2e \label{ineq2}\\
\leq&\sp  2\| H^{0}_{\tcal_{0}} \|_{F} + \| \Hbar_{0,1}\|_{F} +  2e'+2e, \label{ineq3}
\end{align}
where $e=\frac { \| X_{\tcal_{0}^{c}} \|_{2,1}} {\sqrt{k}}$ and $e'= \frac{\| \overline X_{r,k} \|_{*}}{\sqrt{2r}}$.
Inequality \eqref{ineq2} is the consequence of  \eqref{lF-l*}, and for deriving \eqref{ineq3} we use \eqref{H*} and \eqref{norms ineq}. 

Using \eqref{ineq3} and \eqref{H*}, we can continue to upper bound the summation of the Frobenius norms in \eqref{sumFro} as follows:
\begin{align}
\mathcal{S}\leq&\sp   3\| H^{0}_{\tcal_{0}} \|_{F} + \| \Hbar_{0,1}\|_{F} +  4e'+2e.\nonumber \\
\leq&\sp 3\sqrt{2} \| H^{0}_{\tcal_{0}}+ \Hbar_{0,1}\|_{F}+ e'', \label{sumFro2}
\end{align}
where $e'' = 4 e' + 2 e$. Note that for the last inequality we use 
\[\|H^{0}_{\tcal_{0}}\|_{F} +\|\Hbar_{0,1}\|_{F} \leq \sqrt2 \|H^{0}_{\tcal_{0}}+\Hbar_{0,1}\|_{F},\]
which holds because $H^{0}_{\tcal_{0}}$ and $\Hbar_{0,1}$ are orthogonal to each other. 
\newline

Meanwhile, by such specific error decomposition we have 
\[\rank(H^{0}_{\tcal_{0}} + \Hbar_{0,1}+ \Hbar_{1,1}) \leq 6r,\] 
\[\supp(H^{0}_{\tcal_{0}} +  \Hbar_{0,1}+\Hbar_{1,1})\leq 2k,\]
and thus according to Definition \ref{def: RIP} we can write
\if@twocolumn
\begin{align}
&(1- \delta_{6r,2k}) \| H^{0}_{\tcal_{0}} + \Hbar_{0,1}+ \Hbar_{1,1} \|_{F}^{2} \nonumber\\
& \quad\leq  \| \acal(H^{0}_{\tcal_{0}} + \Hbar_{0,1}+ \Hbar_{1,1})\|_{\ell_2}^{2}. \label{RIP-error}
\end{align}
\else
\begin{align}
(1- \delta_{6r,2k}) \| H^{0}_{\tcal_{0}} + \Hbar_{0,1}+ \Hbar_{1,1} \|_{F}^{2} 
 \leq  \| \acal(H^{0}_{\tcal_{0}} + \Hbar_{0,1}+ \Hbar_{1,1})\|_{\ell_2}^{2}. \label{RIP-error}
\end{align}
\fi
On the other hand,
\if@twocolumn
\begin{align}
&\| \acal(H^{0}_{\tcal_{0}} + \Hbar_{0,1}+ \Hbar_{1,1})\|_{\ell_2}^{2} \nonumber\\
&\quad = \Big| \langle \acal(H^{0}_{\tcal_{0}} + \Hbar_{0,1}+ \Hbar_{1,1}), \acal(H - \widetilde H ) \rangle \Big|, \nonumber
\end{align}
\else
\begin{align}
\| \acal(H^{0}_{\tcal_{0}} + \Hbar_{0,1}+ \Hbar_{1,1})\|_{\ell_2}^{2} 
 = \Big| \langle \acal(H^{0}_{\tcal_{0}} + \Hbar_{0,1}+ \Hbar_{1,1}), \acal(H - \widetilde H ) \rangle \Big|, \nonumber
\end{align}
\fi
where,
\[\widetilde H = \sum_{ \substack{ i\geq 0, j\geq 1\\ (i,j) \neq (1,1),(0,1)}}  \Hbar_{i,j}.\]
Using the Cauchy-Schwarz inequality and \eqref{epsilon} yields:
\begin{align}
&\Big| \langle \acal(H^{0}_{\tcal_{0}}+ \Hbar_{0,1} +\Hbar_{1,1}), \acal(H ) \rangle  \Big|  \nonumber\\
&\quad \leq \| \acal(H^{0}_{\tcal_{0}}+\Hbar_{0,1}+\Hbar_{1,1})\|_{\ell_2} \, \| \acal(H)\|_{\ell_2 }\nonumber \\
&\quad \leq  2\varepsilon \, \sqrt{1+\delta_{6r,2k}} \, \| H^{0}_{\tcal_{0}}+\Hbar_{0,1}+\Hbar_{1,1}\|_{F}.
\label{leftineq}
\end{align}
Moreover, because $H^{0}_{\tcal_{0}}+\Hbar_{0,1}$ is at most a rank $4r$, $k$-joint-sparse matrix orthogonal to the summation terms in $\widetilde H$, and so is $\Hbar_{1,1}$ (a rank $2r$ and $k$-joint-sparse matrix), we can use Lemma \ref{RIPlemma} to write
\if@twocolumn
\begin{align}
&\Big| \langle \acal(H^{0}_{\tcal_{0}}+\Hbar_{0,1}+\Hbar_{1,1}), \acal( \widetilde H ) \rangle  \Big| \nonumber\\
&\quad \leq \Big(\delta_{6r,2k} \|H^{0}_{\tcal_{0}} + \Hbar_{0,1}\|_{F} +\delta_{4r,2k} \| \Hbar_{1,1}\|_{F}\Big) \nonumber \\
&\quad \quad \times \Big(\sum_{ \substack{ i\geq 0, j\geq 1\\ (i,j) \neq (1,1),(0,1)}} \| \Hbar_{i,j} \|_{F} \Big) \nonumber \\
&\quad \leq\sqrt2 \,\delta_{6r,2k} \|H^{0}_{\tcal_{0}}+\Hbar_{0,1}+\Hbar_{1,1}\|_{F} \times \mathcal{S} \nonumber \\
&\quad \leq 6\, \delta_{6r,2k} \|H^{0}_{\tcal_{0}} + \Hbar_{0,1}+\Hbar_{1,1} \|^{2}_{F}  \nonumber\\
&\quad \quad + \sqrt2 \,e'' \,\delta_{6r,2k} \|H^{0}_{\tcal_{0}}+\Hbar_{0,1}+\Hbar_{1,1}\|_{F}. \label{rightineq}
\end{align}
\else
\begin{align}
\Big| \langle \acal(H^{0}_{\tcal_{0}}+\Hbar_{0,1}+\Hbar_{1,1}), \acal( \widetilde H ) \rangle  \Big| 
& \leq \Big(\delta_{6r,2k} \|H^{0}_{\tcal_{0}} + \Hbar_{0,1}\|_{F} +\delta_{4r,2k} \| \Hbar_{1,1}\|_{F}\Big)  \times \Big(\sum_{ \substack{ i\geq 0, j\geq 1\\ (i,j) \neq (1,1),(0,1)}} \| \Hbar_{i,j} \|_{F} \Big) \nonumber \\
& \leq\sqrt2 \,\delta_{6r,2k} \|H^{0}_{\tcal_{0}}+\Hbar_{0,1}+\Hbar_{1,1}\|_{F} \times \mathcal{S} \nonumber \\
& \leq 6\, \delta_{6r,2k} \|H^{0}_{\tcal_{0}} + \Hbar_{0,1}+\Hbar_{1,1} \|^{2}_{F}  
+ \sqrt2 \,e'' \,\delta_{6r,2k} \|H^{0}_{\tcal_{0}}+\Hbar_{0,1}+\Hbar_{1,1}\|_{F}. \label{rightineq}
\end{align}
\fi
The second inequality is the outcome of the orthogonality of $H^{0}_{\tcal_{0}}+\Hbar_{0,1}$ and $\Hbar_{1,1}$, and replacing the sum of the Frobenius norms by its upper bound \eqref{sumFro2} gives the last inequality. Inequalities \eqref{leftineq} and \eqref{rightineq} are providing an upper bound on $\| \acal(H^{0}_{\tcal_{0}}+\Hbar_{0,1}+\Hbar_{1,1})\|_{\ell_{2}}^{2}$:
\if@twocolumn
\begin{align}
&\| \acal(H^{0}_{\tcal_{0}}+\Hbar_{0,1}+\Hbar_{1,1})\|_{\ell_2}^{2} \nonumber\\
&\quad \leq 6\, \delta_{6r,2k} \|H^{0}_{\tcal_{0}} + \Hbar_{0,1}+\Hbar_{1,1} \|^{2}_{F}\nonumber \\
 & \quad+ \Big(\sqrt2\, e'' \,\delta_{6r,2k} +2\varepsilon \, \sqrt{1+\delta_{4r,2k}} \Big) \|H^{0}_{\tcal_{0}} + \Hbar_{0,1}+\Hbar_{1,1} \|_{F}.  \nonumber
\end{align}
\else
\begin{align}
\| \acal(H^{0}_{\tcal_{0}}+\Hbar_{0,1}+\Hbar_{1,1})\|_{\ell_2}^{2} 
&\quad \leq 6\, \delta_{6r,2k} \|H^{0}_{\tcal_{0}} + \Hbar_{0,1}+\Hbar_{1,1} \|^{2}_{F}\nonumber \\
 & \qquad + \Big(\sqrt2\, e'' \,\delta_{6r,2k} +2\varepsilon \, \sqrt{1+\delta_{4r,2k}} \Big) \|H^{0}_{\tcal_{0}} + \Hbar_{0,1}+\Hbar_{1,1} \|_{F}.  \nonumber
\end{align}
\fi
Now by using the lower bound derived in \eqref{RIP-error}, we can deduce the following inequality:
\begin{eqnarray}
\|H^{0}_{\tcal_{0}} + \Hbar_{0,1}+\Hbar_{1,1} \|_{F} \leq \alpha e'' + \beta \epsilon, \label{ineq1a}
\end{eqnarray}
where, \[\alpha= \sqrt{2} \frac{ \delta_{6r,2k} } {1-7\delta_{6r,2k} } \] and \[\beta = 2 \frac{ \sqrt{ 1+\delta_{6r,2k}} } {1-7\delta_{6r,2k} }.\] 
\newline
Finally, we upper bound the whole error term $H$ as follows:
\if@twocolumn
\begin{align}
\|H\|_{F} \leq&\sp \|H^{0}_{\tcal_{0}} +\Hbar_{0,1}+ \Hbar_{1,1} \|_{F} + \mathcal{S} \nonumber \\
\nonumber\\
\leq&\sp (3\sqrt2 +1)\|H^{0}_{\tcal_{0}} + \Hbar_{0,1}+\Hbar_{1,1} \|_{F} + e'' \label{final-1} \\
\nonumber\\
\leq&\sp \kappa'_{0} \left(     \frac { \| X_{\tcal_{0}^{c}} \|_{2,1}} {\sqrt{k}} + \frac{\| \overline X_{r,k} \|_{*}}{\sqrt{2r}}   \right) + \kappa'_{1} \varepsilon. \label{final-2}
\end{align}
\else
\begin{align}
\|H\|_{F} \leq&\sp \|H^{0}_{\tcal_{0}} +\Hbar_{0,1}+ \Hbar_{1,1} \|_{F} + \mathcal{S} \nonumber \\
\leq&\sp (3\sqrt2 +1)\|H^{0}_{\tcal_{0}} + \Hbar_{0,1}+\Hbar_{1,1} \|_{F} + e'' \label{final-1} \\
\leq&\sp \kappa'_{0} \left(     \frac { \| X_{\tcal_{0}^{c}} \|_{2,1}} {\sqrt{k}} + \frac{\| \overline X_{r,k} \|_{*}}{\sqrt{2r}}   \right) + \kappa'_{1} \varepsilon. \label{final-2}
\end{align}
\fi

Inequality \eqref{final-1} uses the bound derived in \eqref{sumFro2}. Replacing $\|H^{0}_{\tcal_{0}} +\Hbar_{0,1}+ \Hbar_{1,1} \|_{F}$ by its upper bound in \eqref{ineq1a} results in the last line where the constants $\kappa'_{0} = 4+4(3\sqrt2+1)\alpha$ and $\kappa'_{1}=(3\sqrt2+1)\beta$ are fixed for a given RIP constant $\delta_{6r,2k} < 1/7$. This completes the proof of Theorem \ref{main theorem}-a.
\newline

Proof of Theorem \ref{main theorem}-b: we follow the same steps as for Theorem \ref{main theorem}-a. Here, inequalities \eqref{nuc separation} and \eqref{H*} are the outcome of the nuclear norm minimization in \eqref{P1-2} i.e., $\|\ X \|_{*} \geq \| X+H\|_{*}$, and \eqref{H21} holds since $\gamma = \|X\|_{2,1} \geq \|X+H\|_{2,1}$. Thus, one can easily verifies that all bounds \eqref{lF-l*}-\eqref{final-2} are valid and the solution of \eqref{P1-2} obeys the error bound in \eqref{error bound}.
\\

Proof of Theorem \ref{main theorem}-c: it is also similar to Theorem \ref{main theorem}-a and almost all the previous bounds hold. We just need to develop a new bound on $\frac{1}{\sqrt k} \| \Hbar_{\tcal^{c}_{0}} \|_{2,1}  + \frac{1}{\sqrt{2r}}  \| \Hbar \|_{*} $ \ie inequality \eqref{sumFro}. From minimization \eqref{P1-3} we have
\begin{align}
&\|X\|_{2,1}+\lambda \|X\|_{*} \nonumber \\
& \quad \geq \|X+H\|_{2,1}+\lambda \|X+H\|_{*} \nonumber\\
&\quad \geq \|X_{\tcal_{0}}\|_{2,1} - \|H_{\tcal_{0}}\|_{2,1} - \| X_{\tcal_{0}^{c}}\|_{2,1} + \| H_{\tcal_{0}^{c}}\|_{2,1}  \nonumber \\
&\quad\quad+ \lambda \left(   \| X_{r,k} \|_{*} + \| \Hbar  \|_{*}   - \| \overline X_{r,k} \|_{*} - \| H^{0}_{\tcal_{0}} \|_{*} \right). \nonumber
\end{align}
Thus, we can write the following bound:
\if@twocolumn
\begin{align}
\|H_{\tcal_{0}^{c}}\|_{2,1}  + \lambda \|\Hbar\|_{*} \leq&\sp  \|H_{\tcal_{0}}\|_{2,1} +   2 \| X_{\tcal_{0}^{c}}\|_{2,1} \nonumber \\ 
&+ \lambda \| H^{0}_{\tcal_{0}} \|_{*} +  2\lambda  \| \overline X_{r,k} \|_{*}. \nonumber
\end{align}
\else
\begin{align}
\|H_{\tcal_{0}^{c}}\|_{2,1}  + \lambda \|\Hbar\|_{*} \leq  \|H_{\tcal_{0}}\|_{2,1} +   2 \| X_{\tcal_{0}^{c}}\|_{2,1} + \lambda \| H^{0}_{\tcal_{0}} \|_{*} +  2\lambda  \| \overline X_{r,k} \|_{*}. \nonumber
\end{align}
\fi
Dividing both sides by $\sqrt k$ yields: 
\if@twocolumn
\begin{align}
&\frac{1}{\sqrt k} \| \Hbar_{\tcal_{0}^{c}} \|_{2,1}  + \frac{\lambda}{\sqrt k} \|\Hbar\|_{*}\nonumber \\
&\quad \leq \| H_{\tcal_{0}}\|_{F} + \lambda \sqrt \frac{2r}{k}  \|H^{0}_{\tcal_{0}}\|_{F} + 
 2e+2 \lambda \sqrt \frac{2r}{k} e' \nonumber \\
&\quad \leq \left(2+\lambda \sqrt \frac{2r}{k}\right)  \|H^{0}_{\tcal_{0}}\|_{F} \nonumber \\
&\quad \quad + \| \Hbar_{0,1}\|_{F}   +2e+2 \left(1+\lambda \sqrt \frac{2r}{k} \right) e'. \nonumber
\end{align}
\else
\begin{align}
\frac{1}{\sqrt k} \| \Hbar_{\tcal_{0}^{c}} \|_{2,1}  + \frac{\lambda}{\sqrt k} \|\Hbar\|_{*}
& \leq \| H_{\tcal_{0}}\|_{F} + \lambda \sqrt \frac{2r}{k}  \|H^{0}_{\tcal_{0}}\|_{F} + 
 2e+2 \lambda \sqrt \frac{2r}{k} e' \nonumber \\
& \leq \left(2+\lambda \sqrt \frac{2r}{k}\right)  \|H^{0}_{\tcal_{0}}\|_{F} 
 + \| \Hbar_{0,1}\|_{F}   +2e+2 \left(1+\lambda \sqrt \frac{2r}{k} \right) e'. \nonumber
\end{align}
\fi
Note that the last inequality uses the bound \[ \| H_{\tcal_{0}}\|_{F}  \leq 2\| H^{0}_{\tcal_{0}} \|_{F} + \| \Hbar_{0,1}\|_{F} +  2e',\] which has been previously  developed for obtaining \eqref{ineq3}.
Now, by letting $\lambda = \sqrt{\frac{k}{2r}}$ we achieve the following upper bound:
\if@twocolumn
\begin{align}
 \frac{1}{\sqrt k} \| \Hbar_{\tcal^{c}_{0}} \|_{2,1}  + \frac{1}{\sqrt{2r}}  \| \Hbar \|_{*}  \leq&\sp   3\|H^{0}_{\tcal_{0}}\|_{F} +  \| \Hbar_{0,1} \|_{F} \nonumber\\
 &+ 4e'+2e,\nonumber 
\end{align}
\else
\begin{eqnarray}
 \frac{1}{\sqrt k} \| \Hbar_{\tcal^{c}_{0}} \|_{2,1}  + \frac{1}{\sqrt{2r}}  \| \Hbar \|_{*} \leq    3\|H^{0}_{\tcal_{0}}\|_{F} +  \| \Hbar_{0,1} \|_{F} + 4e'+2e,\nonumber 
\end{eqnarray}
\fi
which gives the same bound as in \eqref{sumFro2} for sum of the Frobenius norms. All other steps and inequalities that are derived for the proof of Theorem \ref{main theorem}-a are valid as well for Theorem \ref{main theorem}-c. Therefore, solving \eqref{P1-3} with $\lambda=\sqrt{\frac{k}{2r}}$ leads to the same error bound as in \eqref{error bound}.

\section{The Parallel Proximal Algorithm}\label{sec: PPXA}

Problems \eqref{P1-1}-\eqref{P1-3} are nonlinearly constrained non-differentiable convex problems. There are numerous methods for solving problems of this kind. We discuss here our approach based on the Parallel Proximal Algorithm (PPXA) proposed in \cite{proximal-splitting}. A remarkable advantage of using this method is the ability of having a parallel implementation. 
In a nutshell, PPXA is an iterative method for minimizing an arbitrary finite sum of lower semi-continuous (l.s.c.) convex functions, and each iteration consists of computing the \emph{proximity} operators of all functions (which can be done in parallel), averaging their results and updating the solution until the convergence. The proximity operator of a l.s.c. function $f(x) :\Rbb^{n} \rightarrow \Rbb$ is defined as $\prox_{f}: \Rbb^{n}\rightarrow  \Rbb^{n}$: 
\begin{eqnarray}\label{prox}
\underset{\widetilde x \in  \Rbb^{n}} {\argmin} f (\widetilde x) + \frac{1}{2}\|x - \widetilde x \|_{\ell_{2}}^{2},
\end{eqnarray}
which may have an analytical expression or can be solved iteratively. 

In this part we rewrite each of the problems \eqref{P1-1}-\eqref{P1-2} as a minimization of the sum of three l.s.c. convex functions 
\begin{eqnarray}\label{P1-sum}
\argmin_{X\in \Rbb^{\na \times \nb}} f_{1}(X) + f_{2}(X) + f_{3}(X).
\end{eqnarray}
The template of the PPXA algorithm that solves this minimization is demonstrated in Algorithm \ref{PPXA}.
In the following we define those functions $f_{i}$ and we derive their corresponding proximity operators. 
Note that because $f_{i}$ is defined on $\Rbb^{\na \times \nb}$, the definition of the proximity operator in \eqref{prox} is naturally extended for matrices by replacing  the $\ell_{2}$ norm with  the Frobenius norm.  Finally, $\beta >0$ (in Algorithm \ref{PPXA}) is a  parameter  controling the speed of convergence.

\begin{algorithm}  [t!]        \label{PPXA}           
\SetAlgoLined \KwIn{ $y$, $\acal$, the regularization parameters ($\varepsilon$ and $\tau$, $\gamma$ or $\lambda$), the convergence parameter $\beta > 0$.} 
\textbf{Initializations:} \\
$n=0$, $X_{0}=\Gamma_{1,0}=\Gamma_{2,0}=\Gamma_{3,0} \in \Rbb^{n_{1}\times n_{2}}$\\

\Repeat
{convergence}{
\For{$(i =1: 3)$ }
{$P_{i,n} = \prox_{3\beta f_{i} }(\Gamma_{i,n})$\\}
$X_{n+1} =(P_{1,n}+P_{2,n}+P_{3,n})/3$\\ 
\For{$(i =1: 3)$ }
{$\Gamma_{i,n+1} = \Gamma_{i,n} + 2X_{n+1} - X_{n} - P_{i,n}$\\}
}
\caption{The Parallel Proximal Algorithm to solve \eqref{P1-1}-\eqref{P1-3}} 
\end{algorithm}

\subsection{Constrained $\lm$ Norm Minimization}\label{sec:LRstJS PPXA}

We can rewrite the minimization problem \eqref{P1-1} in the form \eqref{P1-sum} where   
\[
\text{$f_{1}(X) = \|X\|_{2,1}$ , $f_{2} (X)= i_{\mathcal{B}_{\ell_2}}(X)$, and $f_{3}(X) =  i_{\mathcal{B}_{*}}(X)$. }
\]
The convex sets ${\mathcal{B}_{\ell_2}}, {\mathcal{B}_{*}}  \subset \Rbb^{n_{1}\times n_{2}}$ correspond to the matrices satisfying the constraints $\left\| y - \acal( X) \right\|_{\ell_{2}} \leq \varepsilon$ and $\|X\|_{*} \leq \tau$, respectively. The indicator function $i_{\mathcal{C}}$ of a convex set $\mathcal{C}$ is defined as follows:

\begin{equation}
i_{\mathcal C}(X) = \left\{ 
  \begin{array}{l l}
    0 & \quad \text{if}\,\,X \in {\mathcal C} \nonumber\\\\
    +\infty &  \quad \text{otherwise}\\
  \end{array} \right.
\end{equation}

Standard calculations show that $\forall i \in\{1,\ldots,\na\}$
\begin{equation} \label{prox-lm}
\Big(\prox_{\alpha{f_{1}}}(X) \Big)_{i,.} =    \frac{  \big(\|X_{i,.}\|_{\ell_2}-\alpha\big)_{+}  }{ \|X_{{i,.}}\|_{\ell_2} } \, { X_{i,.} }
\end{equation}
 which is the \emph{soft thresholding} operator applied on the rows of $X$ to shrink their $\ell_{2}$ norms. By definition, the proximal operator of an indicator function $i_{\mathcal C}(X)$ is the  the orthogonal projection of $X$ onto the corresponding convex set $\mathcal C$. Thus,  $\prox_{\alpha f_{2}}(X)$ and $\prox_{\alpha f_{3}}(X)$ are respectively the orthogonal projection of $X$ onto the sets ${\mathcal{B}_{\ell_{2}}}$ and ${\mathcal{B}_{*}}$.  For a general implicit operator $\acal$, the projector onto ${\mathcal{B}_{\ell_{2}}}$ can be computed using an iterative \emph{forward-backward} scheme as proposed in \cite{FadiliS09}. However, once $\acal$ corresponds to a tight frame (\ie $\forall y\in \Rbb^{m}$ $\acal(\acal^{*}(y)) = \nu\, y$ for a constant $\nu$) then according to the \emph{semi-orthogonal linear transform} property of
proximity operators \cite{proximal-splitting}, the orthogonal projection onto ${\mathcal{B}_{\ell_{2}}}$ has the following explicit formulation:
 \begin{eqnarray}\label{prox-l2}
\textit{prox}_{\alpha f_{2}}(X)  =
 X + \nu^{-1}\acal^{*} (\mathbf r)\Big(1- \varepsilon{  \| \mathbf r \|^{-1}_{\ell_{2}} } \Big)_{+} 
\end{eqnarray}
where, $\mathbf{r} = y - \acal(X)$.

Finally, let $X = U \Sigma V^{T}$ be the  singular value decomposition of $X$ with $\Sigma = \text{diag}(\sigma_{1}, \sigma_{2} \ldots)$, then the projection of $X$ onto $\mathcal {B}_{*}$ corresponds to the projection of its singular values onto the following positive simplex ($\ell_{1}$ ball):
\[
\mathcal{B}_{\ell_{1}} = \{ \sigma_{1}, \sigma_{2}, \ldots,  \in \Rbb_{+} : \sum_{i} \sigma_{i} \leq \tau \}.
\] 
More precisely, we have
\[
\prox_{\alpha f_{3}}(X) = U\diag(\overline \sigma_{1} ,\ldots, \overline \sigma_{r}) V^{T},
\]
where $[\overline \sigma_{1}, \overline \sigma_{2}, \ldots] = \mathcal{P}_{\mathcal{B}_{\ell_{1}}} ([\sigma_{1}, \sigma_{2}, \ldots ] )$ and $\mathcal{P}_{\mathcal{B}_{\ell_{1}}}$
 denotes the orthogonal projection of the singular values onto the set $\mathcal{B}_{\ell_{1}}$. This projection can be computed within few low-complex iterations using the $\ell_{1}$ projection Matlab implementation of the SPGL1 software package \cite{spgl1}.


\subsection{Constrained Nuclear Norm Minimization}\label{sec:JSstLR PPXA}

We rewrite the minimization \eqref{P1-2} as \eqref{P1-sum} where   
\[
\text{$f_{1}(X) = \|X\|_{*}$ , $f_{2} (X)= i_{\mathcal{B}_{\ell_2}}(X)$, and $f_{3}(X) =  i_{\mathcal{B}_{\lm}}(X)$. }
\]
The function $f_{2}$ and its proximity operator are the same as in the previous part. The convex set ${\mathcal{B}_{\lm}}$ corresponds to the matrices satisfying  $\|X\|_{2,1} \leq \gamma$, and $\prox_{\alpha f_{3}}(X)$ is the orthogonal projection of $X$ onto ${\mathcal{B}_{\lm}}$.\footnote{The Matlab code to perform such projection can be found in the SPGL1 package and is a natural extension of the vector projection onto the $\ell_{1}$ ball.}

%

The proximity operator of the nuclear norm is the \emph{soft thresholding} shrinkage operator applying on the singular values of $X$:
\begin{eqnarray}\label{prox*}
\prox_{\alpha f_{1}}(X) = U\diag(\overline \sigma_{1} , \overline \sigma_{2} , \ldots) V^{T} 
\end{eqnarray}
where, $\forall i$ $
\overline \sigma_{i} =
    \big(\sigma_{i} - \alpha \big)_{+}$.


\subsection{Constrained, Nuclear norm plus $\lm$ Norm Minimization}\label{sec:LS+JS PPXA}

Here, we rephrase problem \eqref{P1-3} as \eqref{P1-sum} where the functions in the summation are 
\[
\text{$f_{1}(X) =  \|X\|_{2,1}$ , $f_{2} (X)= i_{\mathcal{B}_{\ell_2}}(X)$, and $f_{3}(X) = \lambda \|X\|_{*}$. }
\]
According to our previous descriptions, $\prox_{\alpha f_{2}}$ is the orthogonal projection onto the set ${\mathcal{B}_{\ell_2}}(X)$, $\prox_{\alpha f_{3}}$ is the singular value soft thresholding as in \eqref{prox*} and $\prox_{\alpha f_{1}}$ is equivalent to the soft thresholding of the $\ell_2$ norms of the rows of $X$ as in \eqref{prox-lm}. 


Note that, for all three above-mentioned implementations using the "economized" singular value decomposition is very useful as the data matrix is assumed to be low-rank and thus it avoids a huge computation at each iteration (particularly for a high-dimensional data).\footnote{The "economized" singular value decomposition of a rank $r$ matrix, here, refers to computing only the $r$ largest (\ie nonzero) singular values.} Since the rank $r$ of the solution at each iteration is unknown, at the first iteration of Algorithm \ref{PPXA}, $r$ is initially set to some value and later, incremented by doubling its value until achieving the zero singular value in \ref{sec:LRstJS PPXA}, or the singular value soft thresholding parameter $\alpha$ in \ref{sec:JSstLR PPXA} and \ref{sec:LS+JS PPXA}. For the next iterations, $r$ is initialized by the value of the last iteration and so on. 
This idea has been used in the SVT software package \cite{svt} for the singular values soft thresholding and our experiments indicates that in the early iterations of Algorithm \ref{PPXA} and within at most 2 or 3 sub-iterations, $r$ is set to its correct value, and remains fixed for the forthcoming iterations.

\section{Numerical Experiments}\label{ch2-sec:expe}

In this section we conduct a series of experiments in order to demonstrate the average behavior of the proposed simultaneous low-rank and joint-sparse recovery approach for different problem setups, \ie different matrix sizes $\na\times \nb$, joint-sparsity levels $k$, ranks $r$ and number of CS measurements $m$.  Among the three formulations suggested in this paper, we choose the nuclear norm plus the $\lm$ mixed-norm minimization scheme in \eqref{P1-3} (referred to as $\LRJS$) as it requires the minimum prior knowledge about the underlying data, and the regularization parameter is set to $\lambda=\sqrt{\frac{k}{2r}}$ as suggested by Theorem \ref{main theorem}-c. 
\begin{figure*}[t]
  \centering \subfigure[\LRJS, $r=2$]{
    \label{fig1}
      \centering \includegraphics[width=.45 \linewidth, height = 6.2cm]{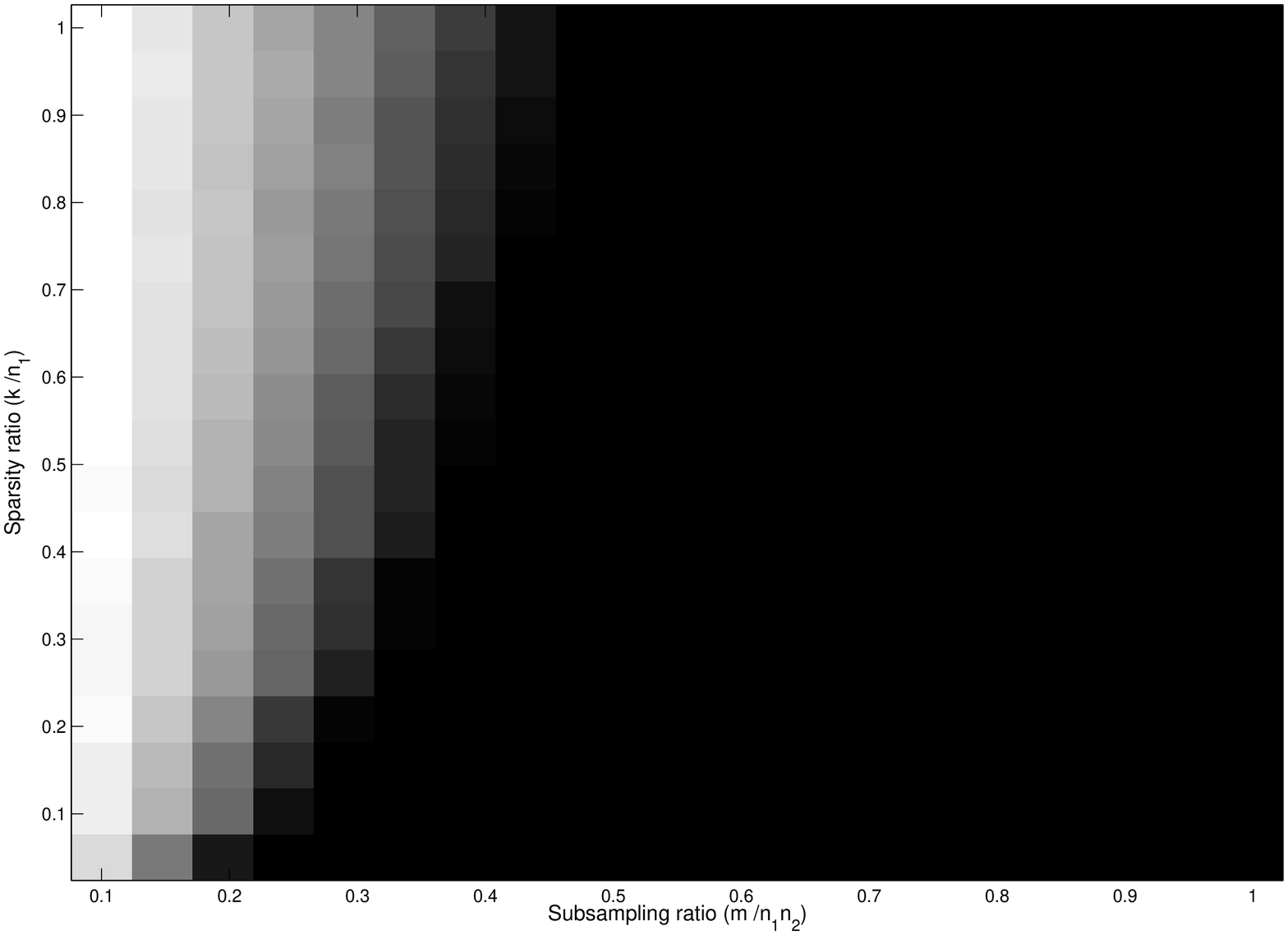}}
    \hfil
     \centering  \subfigure[$\lm$ minimization, $r=2$]{
    \label{fig2}
     \centering \includegraphics[width= .45\linewidth, height = 6.2cm]{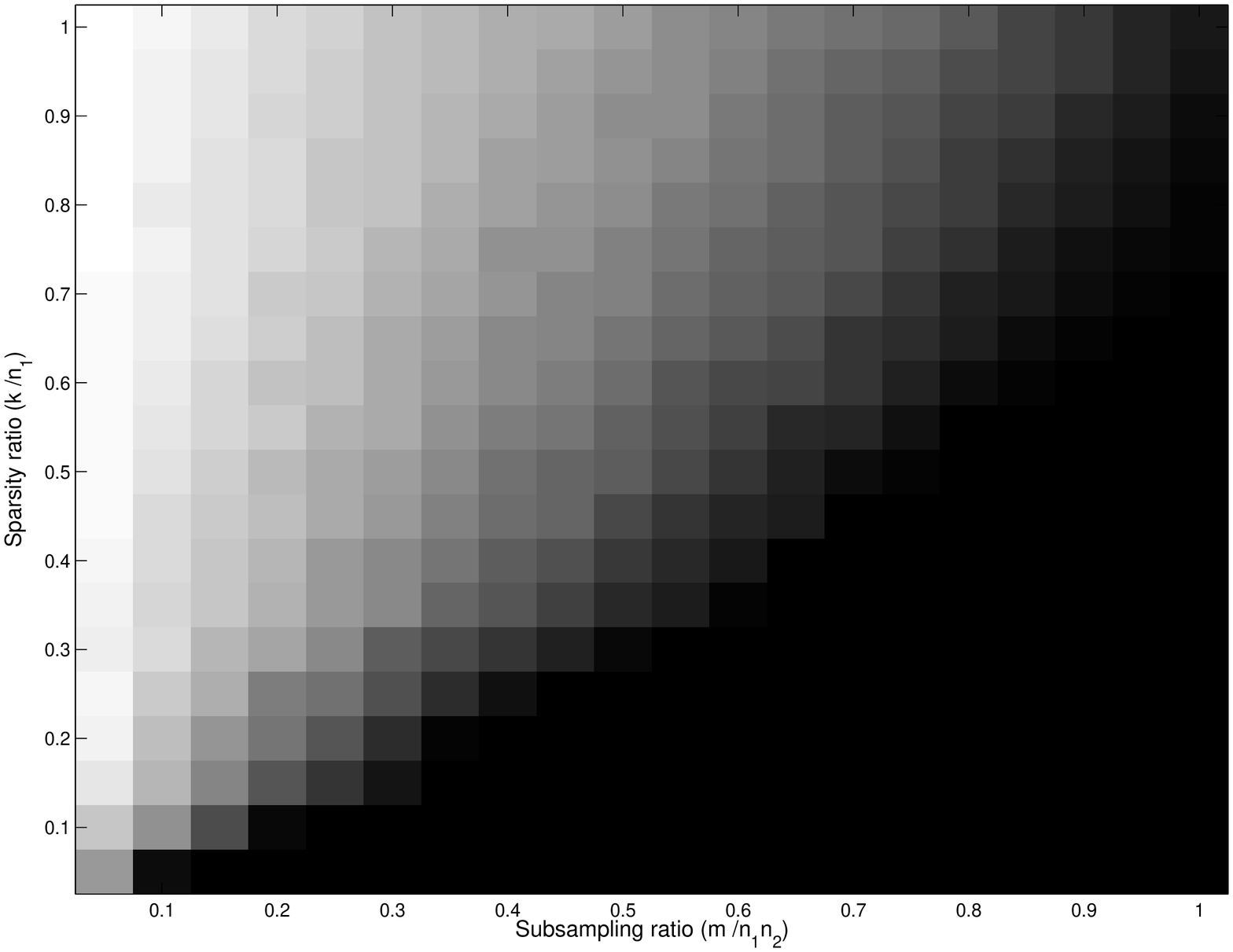}}
        \hfil
     \centering  \subfigure[\LRJS, $r=4$]{
    \label{fig3}
     \centering \includegraphics[width= .45 \linewidth, height = 6.2cm]{Ch3/TIT_PhT_LRJS_r2}}
            \hfil
     \centering  \subfigure[$\lm$ minimization, $r=4$]{
    \label{fig4}
     \centering \includegraphics[width= .45 \linewidth, height = 6.2cm]{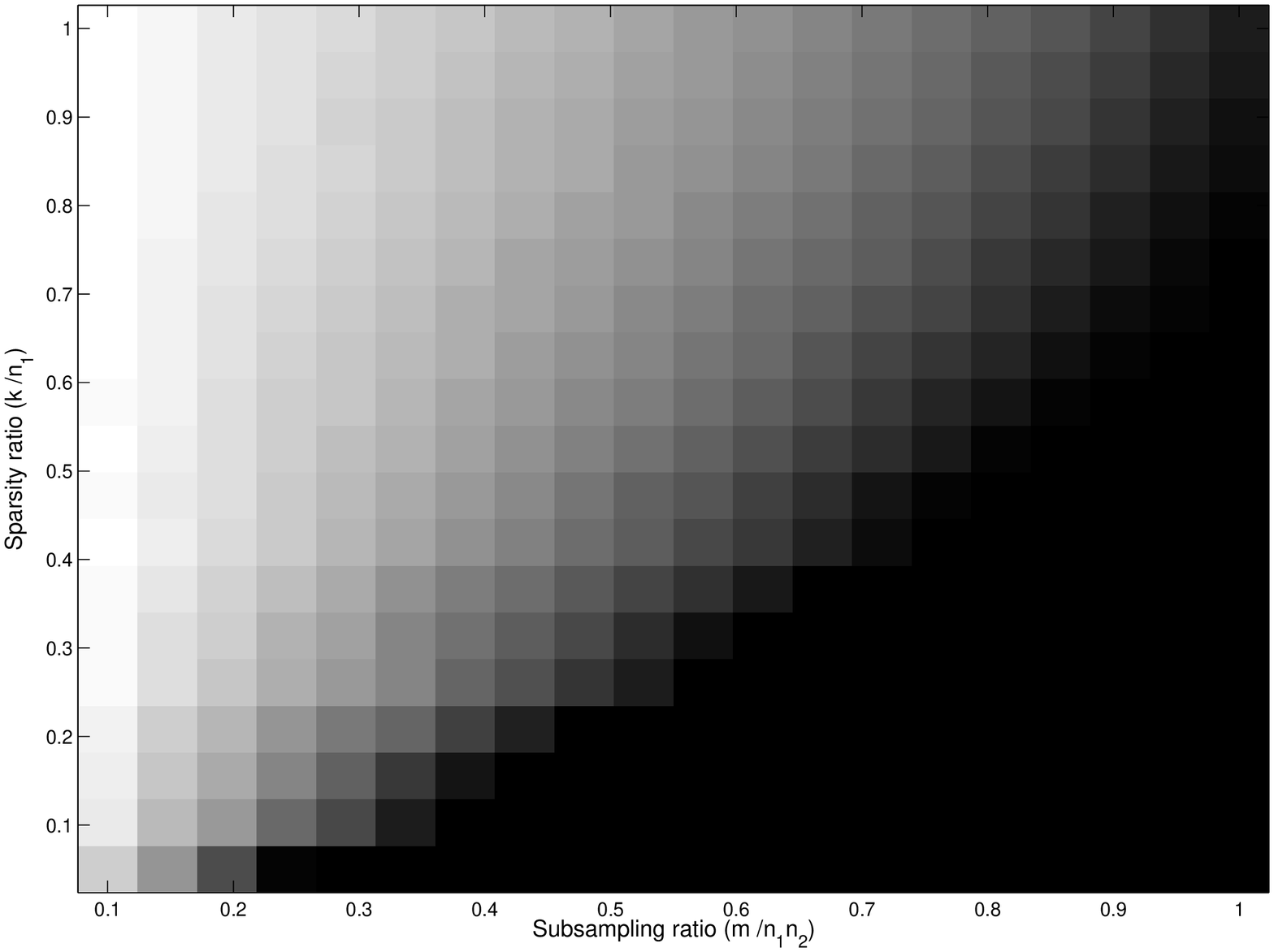}}
  \caption{Reconstruction error phase transitions for $40\times 40$ matrices of rank $r$ acquired by a dense i.i.d. Gaussian sampling matrix: Sparsity ratio $(k/\na)$ vs. Subsampling ratio $(\frac{m}{\na\nb})$.}
\end{figure*}
A rank $r$ and $k$-joint sparse matrix $X\in\Rbb^{\na\times \nb}$ is generated by multiplying two independent i.i.d. Gaussian matrices $X_{1}\in\Rbb^{k\times r}$ and $X_{2}\in\Rbb^{r\times \nb}$, and populating uniformly at random $k$ rows of $X$ by the rows of $X_{1}X_{2}$. The figures that are used for performance evaluation demonstrate the empirical reconstruction error phase transition: for a certain problem setup $(\na, \nb, k, r, m)$, the normalized reconstruction error $\|X-\widehat X\|_{F}/\|X\|_{F}$ has been averaged over twenty independent random realizations of $X$ and $\acal$. Figures are plotted in gray scale and white pixels correspond to high reconstruction errors while black pixels are indicating the recovery regions.

Figure \ref{fig1} demonstrates the performance of our proposed scheme at recovering $40\times 40$ random matrices of rank $r=2$, and for different joint-sparsity ratios ($\frac{k}{\na}$) and subsampling factors ($\frac{m}{\na\nb}$). A dense i.i.d. random Gaussian sampling scheme $\acal$ has been applied for compressed sampling of data. We can observe that for recovering data matrices with high joint-sparsity ratios more CS measurements are required. Meanwhile, it turns out that $\LRJS$ is able to recover even dense (non-sparse) data matrices (\ie, $k=\na$) subsampled by retaining less than $40\%$ of their original sizes. This should not come as a surprise~: despite loosing  sparsity, $\LRJS$ still takes advantage of data correlations due to the low-rank structure, which ultimately results in a successful reconstruction for an oversampling factor nearly four times the degrees of freedom of the original data.

Figure \ref{fig2} illustrates the error phase transition behavior of the $\lm$ minimization \eqref{ch2-l21} for the same setup. As we expect, for such highly structured data, our proposed scheme significantly outperforms the $\lm$ approach by having a much wider recovery region.

Figures \ref{fig3} and \ref{fig4} repeat similar experiments but for rank $r=4$ matrices. The recovery region of the joint \LRJS\ norm minimization is narrowed compared to the Figure \ref{fig1} as more measurements are required to compensate higher degrees of freedom of data in this scenario. However, this method still keeps outperforming the $\lm$ recovery approach.


\begin{figure*}[t!]
  \centering \subfigure[\LRJS, dense sampling]{
    \label{fig5}
      \centering \includegraphics[width= .45 \linewidth, height = 6.2cm]{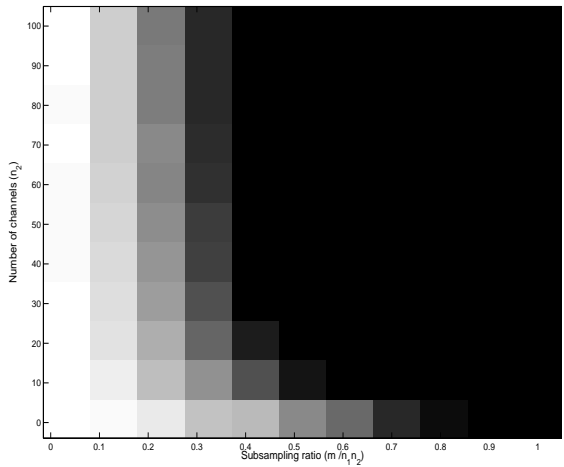}}%
    \hfil
     \centering  \subfigure[\LRJS, distributed independent sampling]{
    \label{fig6}
     \centering \includegraphics[width= .45 \linewidth, height = 6.2cm]{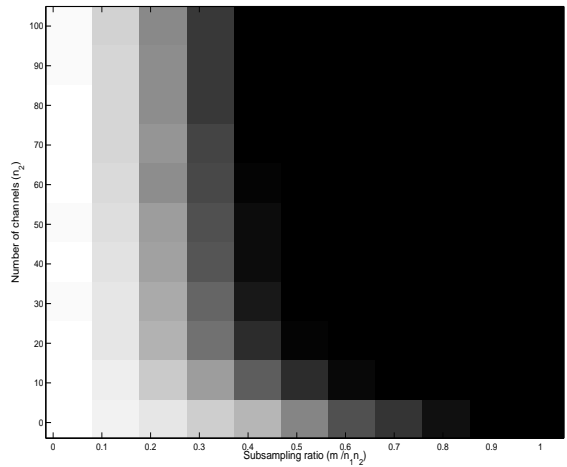}}
        \hfil
     \centering  \subfigure[$\lm$ minimization, dense sampling]{
    \label{fig7}
     \centering \includegraphics[width= .45 \linewidth, height = 6.2cm]{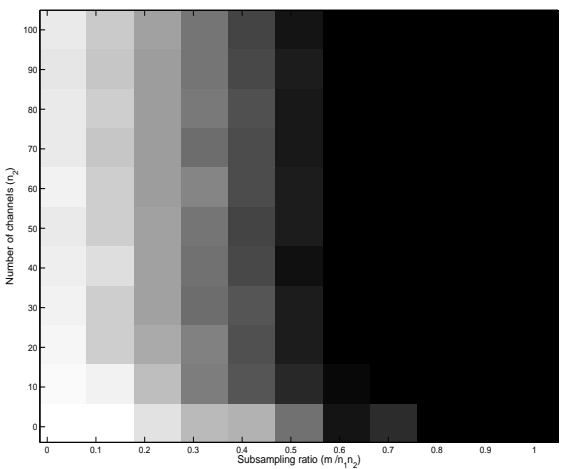}}
            \hfil
     \centering  \subfigure[\LRJS, distributed uniform  sampling]{
    \label{fig8}
     \centering \includegraphics[width= .45 \linewidth, height = 6.2cm]{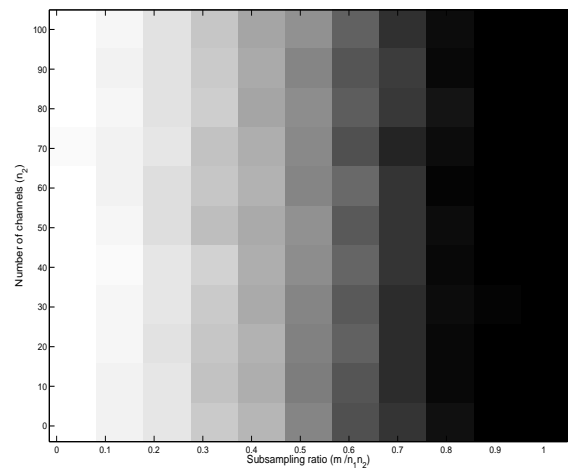}}
  \caption{Reconstruction error phase transitions for $30 \times \nb$ matrices having $r=3$ and $k=10$: Number of columns $(\nb)$ vs. Subsampling ratio $(\frac{m}{\na\nb})$.}
\end{figure*}

\subsection*{Reconstruction Using Distributed Sampling Schemes}
In another experimental setup, we demonstrate the influence of different sampling schemes on our recovery approach.  Besides \emph{dense sampling}, corresponding to the use of dense i.i.d. random Gaussian matrices $A$ in \eqref{sampling model2}, we use the two following schemes particularly suited for distributed sensing applications:
\begin{itemize}
\item \emph{Distributed Uniform Sampling}:
In this case all columns of the data matrix are sampled by a \emph{unique} i.i.d. Gaussian random matrix $\atilde \in \Rbb^{\mhat \times \na}$ and thus the sampling operator $\acal$ corresponds to a block diagonal matrix given below: 
\begin{align}\label{UDS}
{{A}}=\begin{bmatrix}
    \atilde & \mathbf{0}& \ldots
    & \mathbf{0}\\
    \mathbf{0} & \atilde & \ldots
    & \mathbf{0}\\
    \vdots & \vdots & \ddots
    & \vdots \\
    \mathbf{0} & \mathbf{0}& \ldots
    & \atilde
  \end{bmatrix},
\end{align}
where $\mhat = m/\nb$ denotes the number of measurements collected per column.
As previously mentioned in Section \ref{sec: prior-RAJS}, this sensing model is analogous to the multiple measurement vectors (MMV) model.
For multichannel signal applications such as sensor networks where each column of the data matrix represents the observation of a sensor node, this scheme enables all nodes to use a unified mechanism for sampling their own \emph{local} observations in a distributed manner (\ie sampling without intra-channel collaboration). 

\item \emph{Distributed Independent Sampling}:
In this case, an independent random Gaussian matrix $\atilde_{j} \in \Rbb^{\mhat \times \na}$ is used to collect $\mhat$ samples from each column $j$ of the data matrix, and thus the sampling operator corresponds to a block diagonal matrix with distinct blocks:
\begin{align}\label{DBD}
{{A}}=\begin{bmatrix}
    {\atilde_{1}} & \mathbf{0}& \ldots
    & \mathbf{0}\\
    \mathbf{0} & {\atilde_{2}} & \ldots
    & \mathbf{0}\\
    \vdots & \vdots & \ddots
    & \vdots \\
    \mathbf{0} & \mathbf{0}& \ldots
    & {\atilde_{\nb}}
  \end{bmatrix},
\end{align}
where $\mhat = m/\nb$. For multichannel sensing applications with strong intra channel correlations, this  sampling strategy gives the freedom of designing the blocks so that the measurements of different channels carry diverse and non-redundant information to the decoder. 
\end{itemize}


Figures \ref{fig5}-\ref{fig8} demonstrate the reconstruction error phase transitions for various numbers of columns $\nb$ and CS measurements $m$. Data matrices of size $30\times \nb$ with a fixed rank $r=3$ and joint-sparsity level $k=10$ are drawn at random and subsampled by dense, distributed independent or distributed uniform random sampling matrices. These graphs suggest a useful tradeoff for practical multichannel sensing setups, \eg in sensor network applications, in order to design low cost/complex sensors but distribute many of them (high number of channels $\nb$).
As we can see in Figure \ref{fig5}, the recovery region of $\LRJS$ using dense  sampling matrices follows the behavior  predicted by the theoretical analysis (see Corollaries \ref{corollary1} and \ref{main corollary}) \ie,
\[
\frac{m}{\na\,\nb}\gtrsim \O \left( \frac{k\log(\na/k)+kr}{\na\,\nb}+ \frac{r}{\na} \right),
\]
that is a constant plus a decaying term. A very similar recovery performance can be observed in Figure \ref{fig6} using distributed independent sampling. In both cases by increasing $\nb$, $\LRJS$ requires less CS  measurements per column (up to a minimal measurement limit).  

Figure \ref{fig7} demonstrates the performance of the $\lm$ recovery scheme using dense random sampling matrices. Compared to $\LRJS$, the recovery region shrinks and more remarkably for high values of $\nb$, the minimal measurement limit becomes larger \ie more measurements per column is required. 
As a result, for multichannel sensing applications $\LRJS$ achieves better tradeoff between the number of sensor nodes and their complexity compared to $\lm$ recovery approach.


Finally, we repeat the same experiment for the \LRJS\ approach and using distributed uniform random sampling matrices. We can observe in Figure \ref{fig8} that by adding extra channels recovery does not improve much, which is consistent with our early explanations: The uniform sampling is not a proper strategy for CS acquisition of low-rank data matrices because the resulting measurements become highly redundant and they do not carry diverse information across channels. 

%
\section{Conclusion}
In this paper we showed that combining the two key priors in the fields of compressive sampling and matrix completion, namely  joint-sparsity and low-rank structures, can tremendously enhance  CS dimensionality reduction performance. We proposed a joint \LRJS\ convex minimization approach that is able to efficiently exploit such restrictive structure. Our theoretical analysis indicates that the \LRJS\ scheme is robust against measurement noise and non-exact low-rank and joint-sparse data matrices. It can achieve a tight lower bound on the number of CS measurements  improving significantly the state-of-the-art. We provided algorithms to solve this type of minimization problem enabling us to validate these improvements by several numerical experiments.

Extension of this work includes theoretical analysis of the \LRJS\ approach for structured random matrices (\eg, distributed independent random sampling matrices), as well as considering alternative structures such as simultaneous low-rank, sparse (rather than joint-sparse) and positive semidefinite matrices which may find  applications in high-dimensional sparse covariance matrix estimation. We are  also interested in further applications of this approach in multichannel signal compressive sampling.


%

\bibliographystyle{IEEEtran}
\bibliography{Ch4_7}

\end{document}